\documentclass[12pt,letterpaper]{article}

\usepackage{natbib}
\usepackage[ left=1in, top=1in, right=1in, bottom=1in]{geometry}
\usepackage{xcolor,graphicx,bm,bbm,colonequals,amsmath,amssymb,url}
\usepackage{array,tabularx,multirow}
\usepackage{enumitem,algpseudocode}
\usepackage[font={footnotesize}]{caption,subcaption}

\usepackage[normalem]{ulem} 

\bibpunct[, ]{(}{)}{;}{a}{,}{,}

\usepackage{amsthm}
\newtheoremstyle{propstyle} 
    {3mm}                    
    {1mm}                    
    {\itshape}                   
    {}                           
    {\scshape}                   
    {.}                          
    {.5em}                       
    {}  
\theoremstyle{propstyle}
\newtheorem{prop}{Proposition}
\theoremstyle{propstyle}
\newtheorem{definition}{Definition}
\theoremstyle{propstyle}
\newtheorem{lemma}{Lemma}

\newsavebox\ideabox
\newenvironment{idea}
  {\begin{equation}
   \begin{lrbox}{\ideabox}
   \begin{minipage}{\dimexpr\columnwidth-2\leftmargini}
   \setlength{\leftmargini}{0pt}%
   \begin{quote}}
  {\end{quote}
   \end{minipage}
   \end{lrbox}\makebox[0pt]{\usebox{\ideabox}}
   \end{equation}}

\newcommand{\ba}{\mathbf{a}}
\newcommand{\bh}{\mathbf{h}}

\newcommand{\bb}{\mathbf{b}}

\newcommand{\bs}{\mathbf{s}}

\newcommand{\bx}{\mathbf{x}}
\newcommand{\by}{\mathbf{y}}

\newcommand{\bS}{\mathbf{S}}
\newcommand{\bz}{\mathbf{z}}

\newcommand{\bW}{\mathbf{W}}
\newcommand{\bG}{\mathbf{G}}
\newcommand{\bL}{\mathbf{L}}
\newcommand{\bI}{\mathbf{I}}

\newcommand{\bV}{\mathbf{V}}

\newcommand{\bX}{\mathbf{X}}

\newcommand{\bB}{\mathbf{B}}

\newcommand{\bfsigma}{\bm{\Sigma}}
\newcommand{\bfzero}{\mathbf{0}}

\newcommand{\bfgamma}{\bm{\gamma}}

\newcommand{\bftheta}{\bm{\theta}}
\newcommand{\bfeta}{\bm{\eta}}
\newcommand{\bfnu}{\bm{\nu}}

\newcommand{\bfepsilon}{\bm{\epsilon}}

\newcommand{\bfOmega}{\bm{\Omega}}

\newcommand{\im}{{i_1,\ldots,i_m}}

\newcommand{\jm}{{j_1,\ldots,j_m}}

\newcommand{\E}{E}

\newcommand{\cov}{cov}

\newcommand{\GP}{GP}

\newcommand{\blockdiag}{blockdiag}

\newcommand{\knots}{\mathcal{Q}}
\newcommand{\knot}{\mathbf{q}}

\newcommand{\normal}{\mathcal{N}}
\newcommand{\order}{\mathcal{O}}
\newcommand{\modu}{T}
\renewcommand{\prec}{\bm{\Lambda}}
\newcommand{\pprec}{\widetilde{\bm{\Lambda}}}
\newcommand{\domain}{\mathcal{D}}
\newcommand{\pp}{\tau}
\newcommand{\ppm}{\widetilde\tau}
\newcommand{\deltam}{\widetilde\delta}
\newcommand{\locs}{\mathcal{S}}

\usepackage{graphicx}
\newcommand{\indep}{\rotatebox[origin=c]{90}{$\models$}}
 \def\eqd{\,{\buildrel d \over =}\,} 

\title{A class of multi-resolution approximations for large spatial datasets}

\date{}

\author{Matthias Katzfuss\thanks{Department of Statistics, Texas A\&M University} \thanks{Corresponding author: \texttt{katzfuss@gmail.com}} \and Wenlong Gong\footnotemark[1]}

\begin{document}

\maketitle

\begin{abstract}
Gaussian processes are popular and flexible models for spatial, temporal, and functional data, but they are computationally infeasible for large datasets. We discuss Gaussian-process approximations that use basis functions at multiple resolutions to achieve fast inference and that can (approximately) represent any spatial covariance structure. We consider two special cases of this multi-resolution-approximation framework, a taper version and a domain-partitioning (block) version. We describe theoretical properties and inference procedures, and study the computational complexity of the methods. Numerical comparisons and an application to satellite data are also provided. 
\end{abstract}

\subsection*{Keywords}
Basis functions; Gaussian process; predictive process; kriging; satellite data; sparsity.

\section{Introduction \label{sec:intro}}

Gaussian processes (GPs) are highly popular models for spatial data, time series, and functions. They are flexible and allow natural uncertainty quantification, but their computational complexity is cubic in the data size. This prohibits GPs from being used directly for the analysis of many modern datasets consisting of a large number of observations, such as satellite remote-sensing data.

Consequently, many approximations or assumptions have been proposed that allow the application of GPs to large spatial datasets. Some of these approaches are most appropriate for capturing fine-scale structure \citep[e.g.,][]{Furrer2006,Kaufman2008}, while others are more capable at capturing large-scale structure \citep[e.g.,][]{Higdon1998,Mardia1998,Wikle1999,Cressie2008,Katzfuss2009,Katzfuss2010,Katzfuss2011}. \citet{Lindgren2011a} proposed an approximation based on viewing a GP with Mat\'ern covariance as the solution to the corresponding stochastic partial differential equation. Vecchia's method and its extensions \citep[e.g.,][]{Vecchia1988,Stein2004,Datta2016,Katzfuss2017a} are discontinuous and assume the so-called screening effect to hold, meaning that any given observation is conditionally independent from other observations given a small subset of (typically, nearby) observations. 

We propose here a class of multi-resolution-approximation ($M$-RA) of GPs, which allows capturing spatial structure at all scales. The $M$-RA framework is based on an orthogonal decomposition of the GP of interest into processes at multiple resolutions by iteratively applying the predictive process \citep{Quinonero-Candela2005,Banerjee2008}. The process at each resolution has an equivalent representation as a weighted sum of spatial basis functions. For increasing resolution, the number of functions increases while their scale decreases. Unlike other multi-resolution models or wavelets \citep[e.g.][]{Chui1992,Johannesson2007,Cressie2008,Nychka2012}, our $M$-RA automatically specifies the basis functions and the prior distributions of their weights to adapt to the given covariance function of interest, without requiring any restrictions on this covariance function. Thus, in contrast to other approaches, it is clear which covariance structure is approximated by the sum of basis functions in the $M$-RA. 

To achieve computational feasibility within the $M$-RA framework, an approximation of the ``remainder process'' at each resolution using so-called modulating functions is necessary. We consider two special cases: For the $M$-RA taper, the modulating functions are taken to be tapering functions (i.e., compactly supported correlation functions). For increasing resolution, the remainder process is approximated with increasingly restrictive tapering functions, leading to increasingly sparse matrices. In contrast, the $M$-RA-block iteratively splits each region at each resolution into a set of subregions, with the remainder process assumed to be independent between these subregions. This can lead to discontinuities at the region boundaries. A special case of the $M$-RA-block \citep{Katzfuss2015} performed very well in a recent comparison of different methods for large spatial data \citep{Heaton2017}. A further special case of the $M$-RA with only one resolution is given by the full-scale approximation \citep{Snelson2007,Sang2011a,Sang2012}.

The $M$-RA is suitable for inference based on large numbers of observations from a GP, which may be irregularly spaced. We will describe inference procedures that rely on operations on sparse matrices for computational feasibility. The $M$-RA-block can deal with massive datasets with tens of millions of observations or more, as it amenable to parallel computations on modern distributed computing systems. It can be viewed as a Vecchia-type approximation \citep{Katzfuss2017a}, and the approximated covariance matrix is a so-called hierarchical off-diagonal low-rank matrix \citep[e.g.,][]{Ambikasaran2016}. The $M$-RA-taper leads to more general sparse matrices, and thus requires more careful algorithms to fully exploit the sparsity structure, but it has the advantage of not introducing artificial discontinuities.

Relative to the $M$-RA-block in \citet{Katzfuss2015}, the contributions of our paper are the following: We introduce a general framework for $M$-RAs that provides a new, intuitive perspective on this approach. This allows an extension of the $M$-RA-block of \citet{Katzfuss2015} that removes the requirement that knots at the finest resolution correspond to the observed locations, and it enables us to introduce a novel $M$-RA-taper approach that extends the ideas of \citet{Sang2012} to more than one resolution. We provide more insights about the theoretical and computational properties of both versions of the $M$-RA. We also include further implementation details and numerical comparisons.

This article is organized as follows. In Section \ref{sec:MRA}, we first describe an exact orthogonal multi-resolution decomposition of a GP, which then leads to the $M$-RA framework and the two special cases described above by applying the appropriate modulating functions. We also study their theoretical properties. In Section \ref{sec:inference}, we discuss the algorithms necessary for statistical inference using the $M$-RA and provide details of the computational complexity. Numerical comparisons on simulated and real data are given in Sections \ref{sec:simulation} and \ref{sec:application}, respectively. We conclude in Section \ref{sec:conclusions}. All proofs can be found in Appendix \ref{app:proofs}. Additional simulation results can be found in a separate Supplementary Material document. All code will be provided upon publication.


\section{Multi-resolution approximations \label{sec:MRA}}

\subsection{The true Gaussian process \label{sec:general}}

Let $\{y_0(\bs) \!: \bs \in \domain\}$, or $y_0(\cdot)$, be the true spatial field or process of interest on a continuous (non-gridded) domain $\domain \subset \mathbb{R}^d$, $d \in \mathbb{N}^+$. We assume that $y_0(\cdot) \sim \GP(0,C_0)$ is a zero-mean Gaussian process with covariance function $C_0$. We place no restrictions on $C_0$, other than assuming that it is a positive-definite function that is known up to a vector of parameters, $\bftheta$. In real applications, $y_0(\cdot)$ will often not have mean zero, but estimating and subtracting the mean is usually not a computational problem. Once $y_0(\cdot)$ has been observed at a set of $n$ spatial locations $\locs$, the basic goal of spatial statistics is to make (likelihood-based) inference on the parameters $\bftheta$ and to obtain spatial predictions of $y_0(\cdot)$ at a set of locations $\locs^P$ (i.e., to obtain the posterior distribution of $\by_0(\locs^P)$). Direct computation based on the Cholesky decomposition of the resulting covariance matrix requires $\mathcal{O}(n^3)$ time and $\mathcal{O}(n^2)$ memory complexity, which is computationally infeasible when $n \gg 10^4$.

\subsection{Preliminaries \label{sec:preliminaries}}

A multi-resolution approximation with $M$ resolutions ($M$-RA) requires two main ``ingredients'': knots and modulating functions. The multi-resolutional set of knots, $\knots \colonequals \{\knots_0,\ldots,\knots_M\}$, is chosen such that, for all $m=0,1,\ldots,M$, $\knots_m = \{\knot_{m,1},\ldots,\knot_{m,r_m}\}$, is a set of $r_m$ knots, with $\knot_{m,i} \in \domain$. We assume that the number of knots increases with resolution (i.e., $r_0<r_1<\ldots<r_M$). An illustration of such a set of knots in a simple toy example is given in Figure \ref{fig:toyillus}.


The second ingredient is a set of ``modulating functions'' \citep{Sang2011a}, $\modu \colonequals \{\modu_0,\modu_1,\ldots,\modu_M\}$, where $\modu_m: \domain \times \domain \rightarrow [0,1]$ is a symmetric, nonnegative-definite function. In Section \ref{sec:examples} we will consider two specific examples, but for now we merely require that 
$\modu_m(\bs_1,\bs_2)$ is equal to 1 when $\bs_1=\bs_2$, and (exactly) equal to 0 when $\bs_1$ and $\bs_2$ are far apart. Here, the meaning of ``far'' depends on the resolution $m$, in that with increasing $m$, the modulating function should be equal to zero for increasingly large sets of pairs of locations in $\domain$.

Based on these ingredients, we make two definitions:

\begin{definition}[Predictive process]
\label{def:pp}
For a generic Gaussian process $x(\cdot) \sim \GP(0,C)$, define $x^{(m)}(\cdot)$ to be the predictive-process approximation \citep{Quinonero-Candela2005,Banerjee2008} of $x(\cdot)$ based on the knots $\knots_m$:
\[
  x^{(m)}(\bs) \colonequals \E\big(x(\bs) | \bx(\knots_m) \big) = \bb(\bs)'\bfeta, \quad \bs \in \domain,
\]
where $\bb(\bs)' = C(\bs,\knots_m)$ and $\bfeta \sim \normal_{r_m}(\bfzero,\prec^{-1})$, with $\prec = C(\knots_m,\knots_m)$.
\end{definition}
That is, the predictive process is a conditional expectation, and hence a smooth, low-rank approximation of $y(\cdot)$, which can also be written as a linear combination of basis functions \citep[cf.][]{Katzfuss2012}. Further, the remainder $x(\cdot) - x^{(m)}(\cdot) \sim \GP(0,C_R)$ is independent of $x(\cdot)$, with positive-definite covariance function $C_R(\bs_1,\bs_2) = C(\bs_1,\bs_2) - \bb(\bs_1)'\prec^{-1}\bb(\bs_2)$ \citep{Sang2012}.

\begin{definition}[Modulated process]
\label{def:modulated}
For a Gaussian process $x(\cdot) \sim \GP(0,C)$, define $[x]_{[m]}(\cdot)$ to be the ``modulated'' process corresponding to $x(\cdot)$:
\[
 [x]_{[m]}(\cdot) \sim \GP(0,[C]_{[m]}), \textnormal{ where } [C]_{[m]}(\bs_1,\bs_2) = C(\bs_1,\bs_2)\cdot\modu_m(\bs_1,\bs_2), \; \bs_1,\bs_2 \in \domain.
\]
\end{definition}
We see that $x(\cdot)$ and $[x]_{[m]}(\cdot)$ have the same variance structure (because $\modu_m(\bs,\bs)=1$), but $[x]_{[m]}(\cdot)$ has a compactly supported covariance function that is increasingly bad approximation of $C$ as $m$ and the distance between $\bs_1$ and $\bs_2$ increase.

\begin{figure}
	\begin{subfigure}{.45\textwidth}
	\centering
	\includegraphics[width =.98\linewidth]{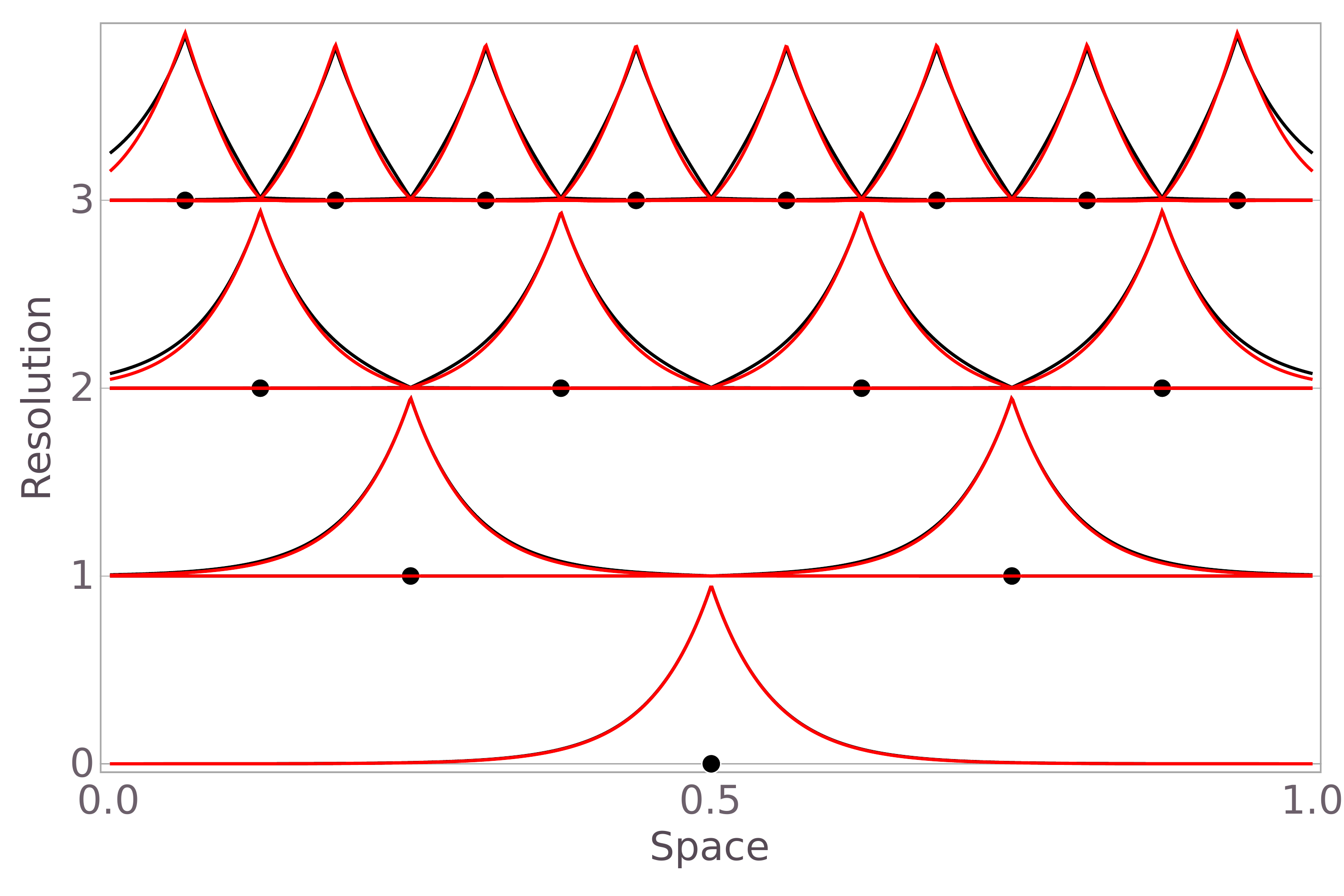}
	\caption{$M$-RA-taper with $d_m = 1.75/2^m$}
	\label{fig:knotsillus}
	\end{subfigure}%
\hfill
	\begin{subfigure}{.45\textwidth}
	\centering
	\includegraphics[width =.98\linewidth]{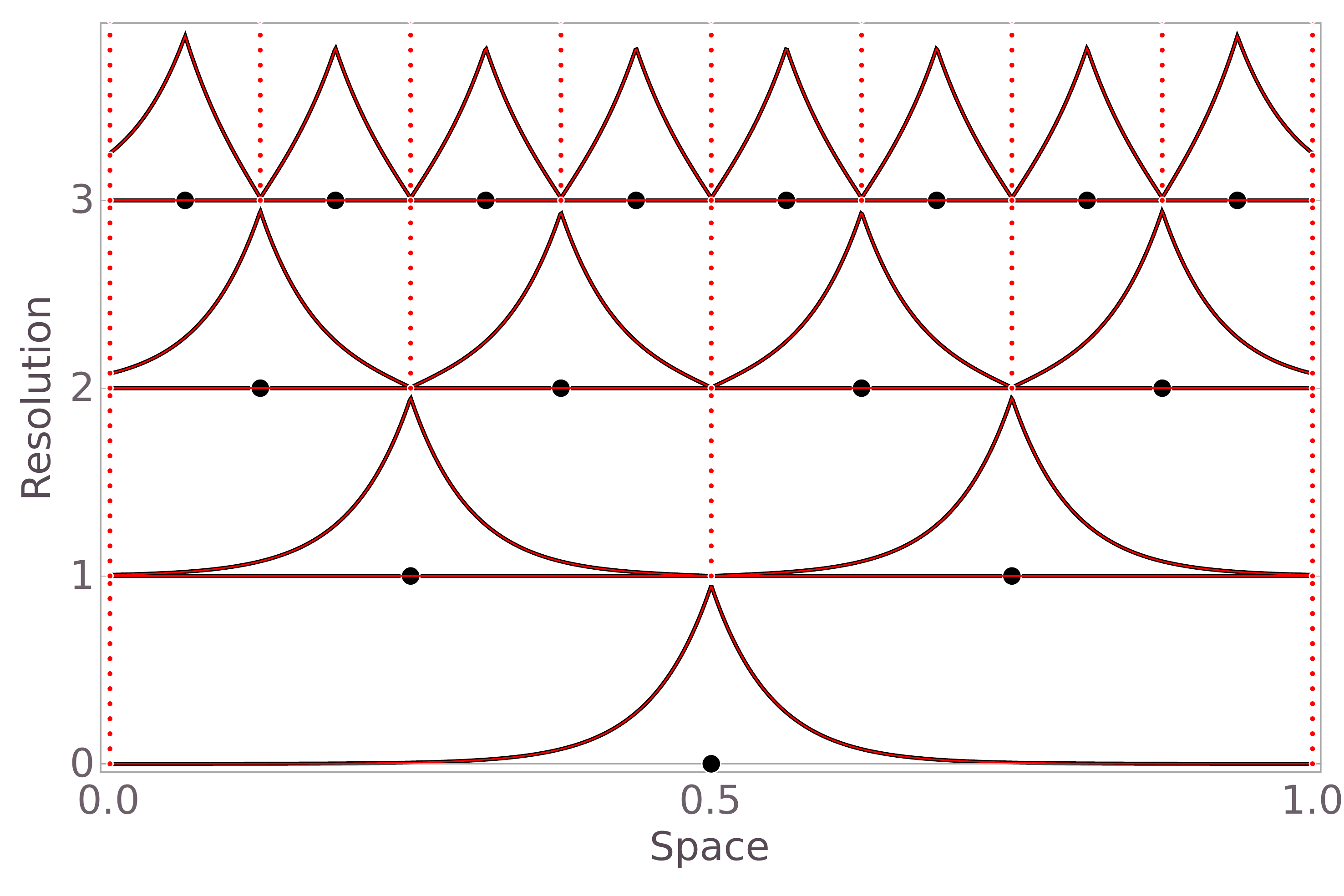}
	\caption{$M$-RA-block}
	\label{fig:bfillus}
	\end{subfigure}%
\hfill
	\begin{subfigure}{.1\textwidth}
		\centering
		\includegraphics[trim = 82mm 47mm 91mm 60mm, clip, width =.6\linewidth]{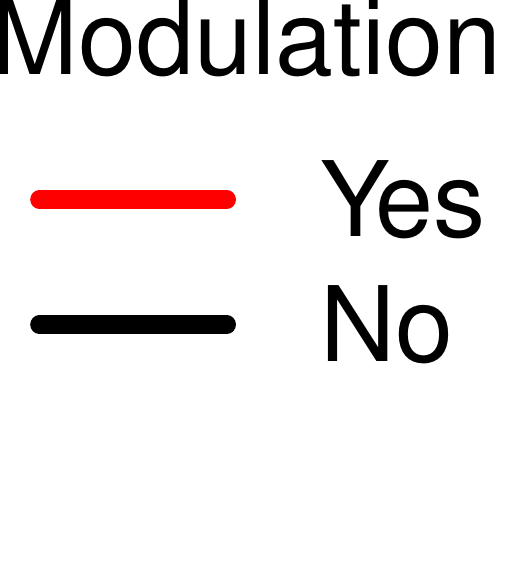}
	\end{subfigure}%
\caption{For $y_0(\cdot) \sim \GP(0,C_0)$ with exponential covariance function $C_0$ on $\domain = [0,1]$, a set of multi-resolution knots (black dots) and the corresponding basis functions using the orthogonal decomposition in \eqref{eq:orthdecomp} (black lines) and using two versions of the $M$-RA (red lines) with $r_0=1$, $J=2$, and $M=3$. The $M$-RA-block is exact in this setting (see Proposition \ref{prop:expexact}), and hence the red and black lines overlap.}
\label{fig:toyillus}
\end{figure}

\subsection{Exact multi-resolution decompositions of Gaussian processes \label{sec:orthdecomp}}

For any Gaussian process $y_0(\cdot) \sim \GP(0,C_0)$ (as specified in Section \ref{sec:general}), using Definition \ref{def:pp}, we can write $y_0(\cdot) = \pp_0(\cdot) + \delta_1(\cdot)$, where $\pp_0(\cdot) \colonequals y_0^{(0)}(\cdot)$ is the predictive process of $y_0(\cdot)$ based on the knots $\knots_0$, and $\delta_1(\cdot) \colonequals y_0(\cdot) - \pp_0(\cdot) \sim \GP(0,w_1)$ is independent from $\pp_0$ and is itself a Gaussian process with (positive-definite) covariance function $w_1$. This allows us to apply again the predictive process to $\delta_1(\cdot)$ (this time based on the knots $\knots_1$) to obtain the decomposition $\delta_1(\cdot) = \pp_1(\cdot) + \delta_2(\cdot)$, and so forth, up to some resolution $M \in \mathbb{N}$. 

This idea enables us to exactly decompose any $y_0(\cdot) \sim \GP(0,C_0)$ into orthogonal (i.e., independent) components at multiple resolutions:
\begin{equation}
\label{eq:orthdecomp}
 y_0(\cdot) \stackrel{d}{=} \pp_0(\cdot) + \ldots + \pp_{M-1}(\cdot) + \delta_M(\cdot),
\end{equation}
where $\pp_m(\cdot)\colonequals \delta_m^{(m)}(\cdot)$ is the predictive process of $\delta_m(\cdot)$ based on knots $\knots_m$, $\delta_0(\cdot)\colonequals y_0(\cdot)$, and $\delta_m(\cdot)\colonequals \delta_{m-1}(\cdot) - \pp_{m-1}(\cdot) \sim \GP(0,w_m)$ for $m=1,\ldots,M$. Further, using the basis-function representation from Definition \ref{def:pp}, we can write each component of the decomposition as $\pp_m(\cdot) = \ba_m(\cdot)'\bfgamma_m$, where $\bfgamma_m \stackrel{ind.}{\sim} \normal_{r_m}(\bfzero,\bfOmega^{-1})$, and starting with $w_0 = C_0$, we have for $m=1,\ldots,M-1$:
\begin{equation}
\label{eq:decompdetails}
\begin{split}
 \ba_m(\bs)' & \colonequals w_m(\bs,\knots_m), \quad \bs \in \domain\\  
 \bfOmega_m & \colonequals w_m(\knots_m,\knots_m)\\
w_{{m+1}}(\bs_1,\bs_2) & \colonequals w_{m}(\bs_1,\bs_2) - \ba_m(\bs_1)'\bfOmega_m^{-1} \ba_m(\bs_2), \quad \bs_1,\bs_2 \in \domain.
\end{split}
\end{equation}
An important feature of this decomposition is that components $\pp_m(\cdot)$ with low resolution $m$ capture mostly smooth, long-range dependence, whereas high-resolution components capture mostly fine-scale, local structure. This is because the predictive process at each resolution $m$ is an approximation to the first $r_m$ terms in the Karhunen-Lo\'eve expansion of $\delta_m(\cdot)$ \citep{Sang2012}. Figure \ref{fig:toyillus} illustrates the resulting basis functions in our toy example. 

It is straightforward to show that the decomposition of the process $y_0(\cdot) \sim \GP(0,C_0)$ in \eqref{eq:orthdecomp} also implies an equivalent decomposition of the covariance function $C_0$:
\begin{equation}
\label{eq:c0decomp}
C_0(\bs_1,\bs_2) = \sum_{m=0}^{M-1} w_m(\bs_1,\knots_m)w_m(\knots_m,\knots_m)^{-1}w_m(\knots_m,\bs_2) + w_M(\bs_1,\bs_2), \quad \bs_1,\bs_2 \in \domain.
\end{equation}

\subsection{The multi-resolution approximation \label{sec:mradef}}

The multi-resolution approximation ($M$-RA) is a ``modulated'' version of the exact decomposition in \eqref{eq:orthdecomp}, which at each resolution $m$ modulates the remainder using the function $\modu_m$ from Section \ref{sec:preliminaries}.
The key idea is that the predictive processes at low resolutions pick up the low-frequency variation in $y_0(\cdot)$, resulting in remainder terms that exhibit variability on smaller and smaller scales as $m$ increases, and so approximating the remainder using more and more restrictive modulating functions causes little approximation error.

\begin{definition}[Multi-resolution approximation ($M$-RA)]
\label{def:mra}
For a given $M \in \mathbb{N}$, the $M$-RA of a process $y_0(\cdot) \sim \GP(0,C_0)$ based on a set of knots $\knots = \{\knots_0,\ldots,\knots_M\}$ and a set of modulating functions $\modu = \{\modu_0,\ldots,\modu_M\}$, is given by
\begin{equation}
\label{eq:mra1}
  y_M(\cdot) = \sum_{m=0}^M \ppm_m(\cdot) = \sum_{m=0}^M \bb_m(\bs)'\bfeta_m,
\end{equation}
where $\ppm_m(\cdot) \colonequals \deltam_m^{(m)}(\cdot)$ and $\bfeta_m \stackrel{ind.}{\sim} \normal_{r_m}(\bfzero,\prec^{-1}_m)$ for $m=0,\ldots,M$; $\deltam_0(\cdot) \colonequals [y_0]_{[0]}(\cdot) \sim \GP(0,v_0)$ with $v_0 = [C_0]_{[0]}$; $\deltam_m(\cdot)=[\deltam_{m-1} - \ppm_{m-1}]_{[m]}(\cdot) \sim \GP(0,v_m)$ for $m=1,\ldots,M$; and 
\begin{equation}
\label{eq:mradetails}
\begin{split}
 \bb_m(\bs)' & \colonequals v_m(\bs,\knots_m), \quad \bs \in \domain, \quad m=0,\ldots,M,\\  
 \prec_m & \colonequals v_m(\knots_m,\knots_m), \quad m=0,\ldots,M,\\
v_{{m+1}}(\bs_1,\bs_2) & \colonequals \big(v_{m}(\bs_1,\bs_2) - \bb_m(\bs_1)'\prec_m^{-1} \bb_m(\bs_2)\big)\cdot\modu_{m+1}(\bs_1,\bs_2), \quad \bs_1,\bs_2 \in \domain, m=0,\ldots,M-1.
\end{split}
\end{equation}
\end{definition}

Figure \ref{fig:toyillus} shows the $M$-RA basis functions in our toy example. As can be seen, the $M$-RA is similar to a wavelet model, in that for increasing resolution $m$, we have an increasing number of basis functions with increasingly compact support. However, in contrast to wavelets, the basis functions $\bb(\cdot)$ and the precision matrix $\prec$ of the corresponding weights in the $M$-RA adapt to the covariance function $C_0$. Defining the basis functions recursively allows the $M$-RA to approximate $C_0$, while in other approaches \citep[e.g., wavelets, or][]{Nychka2012} with explicit expressions for the basis functions, the resulting covariance is less clear.

For ease of notation, we often stack the basis functions as $\bb(\cdot) \colonequals \big(\bb_0(\cdot)',\ldots,\bb_M(\cdot)'\big)'$ and the corresponding coefficients, $\bfeta \colonequals \big(\bfeta_0',\ldots,\bfeta_M'\big)'$, so that 
\begin{equation}
\label{eq:stacked}
y_M(\cdot) = \bb(\cdot)'\bfeta, \qquad \textnormal{where} \; \bfeta \sim \normal_{r}(\bfzero,\prec^{-1}),
\end{equation}
with $\prec \colonequals \blockdiag(\prec_0,\ldots,\prec_M)$ and $r = \sum_{m=0}^M r_m$.

\subsection{Specific examples \label{sec:examples}}

As described in Section \ref{sec:preliminaries}, the $M$-RA requires the choice of two ingredients: knots and modulating functions. In light of the computational complexities discussed in Sections \ref{sec:blockinf}--\ref{sec:details} below, we introduce a factor $J$, often chosen to be equal to 2 or 4. Then, starting with some (small) number of knots $r_0$ at resolution $m=0$, we henceforth assume $r_{m} = Jr_{m-1}$ for $m=1,\ldots,M$.

Regarding the modulating functions, we will now discuss two choices that lead to two important versions of the $M$-RA.

\subsubsection{$M$-RA-block \label{sec:block}}

To define the $M$-RA-block, we need a recursive partitioning of the spatial domain $\domain$, in which each of $J$ regions, $\domain_{1},\ldots,\domain_J$, is again divided into $J$ smaller subregions, and so forth, up to level $M$:
\[
\domain_{j_1,\ldots,j_{m-1}} = \dot{\textstyle\bigcup}_{j_{m}=1,\ldots,J} \, \domain_{j_1,\ldots,j_{m}}, \quad j_1,\ldots,j_{m-1}=1,\ldots,J; \quad m=1,\ldots,M.
\]
We then assume for each resolution $m$ that the modulated remainder $\delta_m(\cdot)$ is independent across partitions at the $m$th resolution. That is, the modulating function is defined as
\begin{equation}
\label{eq:partitioning}
\modu_m(\bs_i,\bs_j) = \begin{cases} 1, & (\im) = (\jm),\\
0, & \textnormal{otherwise}, \end{cases}  \qquad \bs_i \in \domain_\im, \; \bs_j \in \domain_\jm.
\end{equation}
Simply speaking, we have $\modu_m(\bs_1,\bs_2) = 1$ if $\bs_1,\bs_2$ are in the same region $\domain_\jm$, and $\modu_m(\bs_1,\bs_2) = 0$ otherwise. At resolution $m$, $\domain$ is split into $J^m$ subregions. Typically, we assume that the knots at each resolution are roughly equally spread throughout the domain, so that there are roughly the same number $r_m/J^m = r_0$ of knots in every such region.

The $M$-RA-block and the corresponding domain partitioning are illustrated in a toy example in Figure \ref{fig:bfillus}.
The $M$-RA-block was first proposed in \citet{Katzfuss2015} with the restriction that $\knots_M = \locs$. Another special case for $M=1$ is the block-full-scale approximation \citep{Snelson2007,Sang2011a}.

\subsubsection{$M$-RA-taper \label{sec:taper}}

We can also specify the modulating functions to be compactly supported correlation functions, often refered to as tapering functions. For simplicity, we assume here that the modulating functions are of the form,
\[
\modu_m(\bs_1,\bs_2) = \modu_*(\|\bs_1 - \bs_2\|/d_m),
\]
with $d_{m+1} = d_m/J^{1/d}$, where $d$ is the dimension of $\domain$, $\|\cdot\|$ is some norm on $\domain$, and $\modu_*$ is a compactly supported correlation function that is scaled such that $\modu^*(x) = 0$ for all $x \geq 1$. For simplicity, we will use Kanter's function \citep{Kanter1997} in all data examples:
\[
\modu_*(x) \colonequals \begin{cases} 1, & x=0,\\
\textstyle (1-x) \frac{\sin(2 \pi x )}{2 \pi x} + \frac{1 - \cos( 2 \pi x)}{2 \pi^2 x}, & x \in (0, 1),\\
0, & x\geq 1.\end{cases}
\]
For other possible choices of tapering functions, see \citet{Gneiting2002}. The taper-$M$-RA is illustrated in Figure \ref{fig:knotsillus}. A special case of the $M$-RA-taper for $M=1$ is the taper-full-scale approximation \citep{Sang2012,Katzfuss2012}.

\subsection{Properties of the $M$-RA process}

Throughout this subsection, let $y_M(\cdot)$ be the $M$-RA (as described in Definition \ref{def:mra}) of $y_0(\cdot) \sim \GP(0,C_0)$ on domain $\domain$ based on knots $\knots = \{\knots_0,\ldots,\knots_M\}$ and modulating functions $\modu = \{\modu_0,\ldots,\modu_M\}$. All proofs are given in Appendix \ref{app:proofs}.

\begin{prop}[Distribution of the $M$-RA]
\label{prop:mracov}
The $M$-RA is a Gaussian process, $y_M(\cdot) \sim \GP(0,C_M)$, with covariance function 
\[
C_M(\bs_1,\bs_2) = \sum_{m=0}^{M} v_m(\bs_1,\knots_m)v_m(\knots_m,\knots_m)^{-1}v_m(\knots_m,\bs_2), \quad \bs_1,\bs_2 \in \domain,
\]
where $v_m$ is defined in \eqref{eq:mradetails}.
We call $C_M$ the $M$-RA of the covariance function $C_0$.
\end{prop}

\begin{prop}[Duplication of knots]
\label{prop:knots}
If $\knot \in \knots_m$, then $v_{m+l}(\knot,\bs) = 0$ for any $\bs \in \domain$ and $l \geq 1$.
\end{prop}
This proposition implies that there is no benefit to designate the same locations as knots at multiple resolutions; that is, all knot locations in $\knots$ should be unique.

\begin{prop}[Exact variance]
\label{prop:variance}
If $\bs \in \knots$, then the $M$-RA variance at location $\bs$ is exact; that is, $C_M(\bs,\bs)=C_0(\bs,\bs)$.
\end{prop}
This proposition implies that, in contrast to other recent basis-function approaches \citep[e.g.,][]{Lindgren2011a,Nychka2012}, no variance or ``edge'' correction is needed for the $M$-RA if we place a knot location at each observed and prediction location.

Smoothness (i.e., differentiability) is a very important concept in spatial statistics, which has led to the popularity of the Mat\'ern covariance class with a parameter that flexibly regulates differentiability \citep[e.g.,][]{Stein1999}. The following proposition shows that any desired smoothness can be preserved when applying the $M$-RA: 
\begin{prop}[Smoothness]
\label{prop:smoothness}
If $\by_0(\cdot)$ is exactly $p$ times (mean-square) differentiable at $\bs \in \knots$, where $p \in \mathbb{Z}_{\geq 0}$, then $\by_M(\cdot)$ is also exactly $p$ times differentiable at $\bs$, provided that $C_0(\cdot,\knot)$ and $\modu_m(\cdot,\knot)$ are at least $2p$ times differentiable at $\bs$, for any $\knot \in \knots$ and $m=1,\ldots,M$.
\end{prop}
Many commonly used covariance functions (e.g., Mat\'{e}rn) are infinitely differentiable away from the origin. If $C_0$ is such a covariance function, the $M$-RA-block thus has the same smoothness as the original process $\by_0(\cdot)$ at any $\bs$ that is not located on the boundary between subregions at any resolution \citep[cf.][]{Katzfuss2015}. Tapering functions are often smooth away from the origin, except at the distance at which they become exactly zero. Thus, the $M$-RA-taper will typically have the same smoothness at $\bs$ as $\by_0(\cdot)$ if $\modu$ is at least $2p$ times differentiable at the origin and $\bs$ is not exactly at distance $d_m$ from any $\knot \in \knots_m$, for all $m=1,\ldots,M$. Note that this result does not require the smoothness of $y_0$ to be the same at all locations $\bs$; if the smoothness (or other local characteristics) of the covariance function $C_0$ varies over space, the $M$-RA will automatically adapt to this nonstationarity and vary over space accordingly.

There is, however, an issue with the continuity of the $M$-RA-block process at the region boundaries, which can be highly undesirable in prediction maps:
\begin{prop}[Continuity]
\label{prop:continuity}
Assume that $C_0$ is a continuous function. Then, for the $M$-RA-taper, realizations of the corresponding process $y_M(\cdot)$ and the posterior mean (i.e., kriging prediction) surface $\mu_M(\bs) \colonequals \E(y_M(\bs)|\bz)$ based on observations $\bz$ as in \eqref{eq:obs} are both continuous, assuming that $\modu_m$ is continuous for all $m=0,1,\ldots,M$. In contrast, for the $M$-RA-block, $y_M(\cdot)$ and $\mu_M(\cdot)$ are both discontinuous in general at any $\bs$ on the boundary between any two subregions.
\end{prop}

\begin{prop}[Exactness of $M$-RA-block]
\label{prop:expexact}
Let $C_0$ be a (stationary) exponential covariance function on the real line, $\domain=\mathbb{R}$. Further, let $C_M$ be the covariance function of the corresponding $M$-RA-block (see Section \ref{sec:block}) with $r_m=(J-1) J^m$ knots for $m=0,\ldots,M-1$, which are placed such that at each resolution $m$, a knot is located on each boundary between two subregions at resolution $m+1$.
Then, the $M$-RA is exact at every knot location; that is, $C_M(\bs_1,\bs_2) = C_0(\bs_1,\bs_2)$ for any $\bs_1,\bs_2 \in \knots$.
\end{prop}
This proposition is illustrated in Figure \ref{fig:bfillus}.
As we will see in Section \ref{sec:blockinf}, this result allows us to exactly decompose a $n \times n$ exponential covariance matrix in terms of a sparse matrix with $n$ rows but only about $\log_2 n$ nonzero elements per row with $r_0=1$ and $J=2$. This leads to tremendous computational savings (e.g., $\log_2(n) < 30$ for $n=1$ billion).

While the exact result in Proposition \ref{prop:expexact} relies on the Markov property and the exact screening effect of the exponential covariance function (which is a Mat\'{e}rn covariance with smoothness parameter $\nu=0.5$), similar but approximate results are expected to hold for larger smoothness parameters in one dimension. Specifically, \citet{Stein2011b} shows that an asymptotic screening effect holds for $\nu=1.5$ when using conditioning sets of size 2, and he conjectures that an asymptotic screening effect holds for any $\nu$ when using conditioning sets of size greater than $\nu$. This conjecture is also explored numerically in \citet{Katzfuss2017a}. To exploit this screening effect using the $M$-RA-block, we can simply place $c>\nu$ knots near every subregion boundary (i.e., $r_0 = c(J-1)$).


\section{Inference \label{sec:inference}}

In this section, we describe inference for the $M$-RA, based on a set of $n$ measurements at locations $\locs$. We assume additive, independent measurement error, such that 
\begin{equation}
\label{eq:obs}
\bz = \by_M(\locs) + \bfepsilon, \quad \bfepsilon \sim \normal_n(\bfzero,\bV_\epsilon),
\end{equation}
where $\bV_\epsilon$ is a diagonal matrix. We assume that $C_0$ and $\bV_\epsilon$ are fully determined by the parameter vector $\bftheta$, which will be assumed fixed at a particular value, unless noted otherwise.
For the sparsity and complexity calculations, we assume $r_m=r_0J^m$ and $n= \order(r_M)$.

\subsection{General inference results}

\subsubsection{Prior matrices \label{sec:prior}}

For a given set of parameters, the covariance function $C_0$, and hence the basis functions $\bb(\cdot)$ and the precision matrix $\prec$ in \eqref{eq:stacked} are fixed. The prerequisite for inference is to calculate the prior matrices $\prec$ and $\bB\colonequals [\bB_0, \ldots, \bB_M] \colonequals [\bb_0(\locs),\ldots,\bb_M(\locs)]$. Define $\bW_{m,l}^k \colonequals v_k(\knots_m,\knots_l)$ and $\bW_{\locs,m}^k \colonequals v_k(\locs,\knots_m)$, so that $\prec_m = \bW_{m,m}^m$ and $\bB_m=\bW_{\locs,m}^m$. For $m = 0,\ldots,M$, starting with $\bW_{m,l}^0 = v_0(\knots_m,\knots_l)$ and $\bW_{\locs,m}^0 = v_0(\locs,\knots_m)$, it is straightforward to verify that
\begin{equation}
\label{eq:priorw1}
\textstyle\bW_{m,l}^{k+1} = \big( \bW_{m,l}^{k} - \bW_{m,k}^{k} \prec_{k}^{-1} \bW_{l,k}^{k}{}' \big) \circ \modu_{k+1}(\knots_m,\knots_l), \qquad k=0,\ldots,l-1; \quad l = 0,\ldots,m;
\end{equation}
and
 \begin{equation}
\label{eq:priorw2}
\textstyle\bW_{\locs,m}^{k+1} = \big( \bW_{\locs,m}^{k} - \bW_{\locs,k}^{k} \prec_{k}^{-1} \bW_{m,k}^{k}{}' \big) \circ \modu_{k+1}(\locs,\knots_m), \qquad k=0,\ldots,m-1.
\end{equation}
Here, $\circ$ denotes the Hadamard or element-wise product.
Note that $\prec_m$ and $\bB_m$ both grow in dimension and become increasingly sparse with increasing resolution $m$. We have $(\prec_m)_{i,j}=0$ if $\modu_m(\knot_{m,i},\knot_{m,j})=0$, and $(\bB_m)_{i,j} =0$ if $\modu_m(\bs_i,\knot_{m,j})=0$.

\subsubsection{Posterior inference \label{sec:posterior}}

Once $\prec$ and $\bB$ have been obtained, the posterior distribution of the unknown weight vector, $\bfeta$, is given by well-known formulas for conjugate normal-normal Bayesian models:
 \begin{equation}
\label{eq:postdist}
\bfeta \,|\, \bz \sim \normal_{r}(\widetilde\bfnu,\pprec^{-1}),
\end{equation}
where $\pprec = \prec + \bB'\bV_\epsilon^{-1}\bB$, $\widetilde\bfnu = \pprec^{-1} \widetilde\bz$, and $\widetilde\bz = \bB'\bV_\epsilon^{-1}\bz$.

Based on this posterior distribution of $\bfeta$, the likelihood can be written as \citep[e.g.,][]{Katzfuss2014}:
\begin{equation}
\label{eq:likelihood}
 -2 \log L(\bftheta) = -\log |\prec| + \log|\pprec| + \log | \bV_\epsilon | + \bz'\bV_\epsilon^{-1}\bz - \widetilde\bz'\pprec^{-1}\widetilde\bz.
\end{equation}
Using this expression, the likelihood can be evaluated quickly for any given value of the parameter vector $\bftheta$. This allows us to carry out likelihood-based inference (e.g., maximum likelihood or Metropolis-Hastings) on the parameters in $C_0$ and $\bV_\epsilon$, by computing the quantities in \eqref{eq:priorw1}--\eqref{eq:likelihood} for each parameter value.

To obtain spatial predictions for fixed parameters $\bftheta$, note that $\by_M(\locs^P) = \bB^P \bfeta$, where $\bB^P \colonequals \bb(\locs^P)$. Defining $\bW_{\locs^P,l}^{k} \colonequals v_k(\locs^P,\knots_l)$, $\bB^P = [\bB^P_0,\ldots,\bB^P_M]$ can be obtained based on the quantities from Section \ref{sec:prior} by calculating $\bW_{\locs^P,m}^{0} = v_0(\locs^P,\knots_m)$ and 
\[
\textstyle\bW_{\locs^P,m}^{k+1} = \big( \bW_{\locs^P,m}^{k} - \bW_{\locs^P,k}^{k} \prec_{k}^{-1} \bW_{m,k}^{k}{}' \big) \circ \modu_{k+1}(\locs^P,\knots_m), \qquad k=0,\ldots,m-1,
\]
and setting $\bB^P_m = \bW_{\locs^P,m}^{m}$, for $m=0,\ldots,M$.
The posterior predictive distribution is given by,
\begin{equation}
\label{eq:prediction}
 \by_M(\locs^P)\,|\, \bz \sim \normal_{n_P}(\bB^P \widetilde\bfnu, \bB^P \pprec^{-1} \bB^P{}').
\end{equation}

Hence, the main computational effort required for inference is the Cholesky decomposition of $\pprec$, the posterior precision matrix of the basis-function weights in (\ref{eq:postdist}). As $\prec$ and $\bB$ are both sparse, $\pprec$ is a sparse matrix that can be decomposed quickly. Specifically, $\pprec$ has the block structure $\pprec = (\pprec_{m,l})_{m,l=0,\ldots,M}$, where $\pprec_{m,l} = \prec_m \mathbbm{1}_{\{m=l\}} + \bB_m'\bV_\epsilon^{-1}\bB_l$ is an $r_m \times r_l$ matrix whose $(i,j)$th element is 0 if $\not\exists \bs \in \domain$ such that $\modu_m(\knot_{m,i},\bs) \neq 0$ and $\modu_l(\knot_{l,j},\bs)\neq 0$. Figure \ref{fig:sparsity} shows the sparsity structures of $\bB$, $\prec$, and $\pprec$ corresponding to the toy example in Figure {\ref{fig:toyillus}}.

\begin{figure}
	\label{fig:sparsity}
	\begin{subfigure}{.32\textwidth}
		\centering
		\includegraphics[width=.8\linewidth]{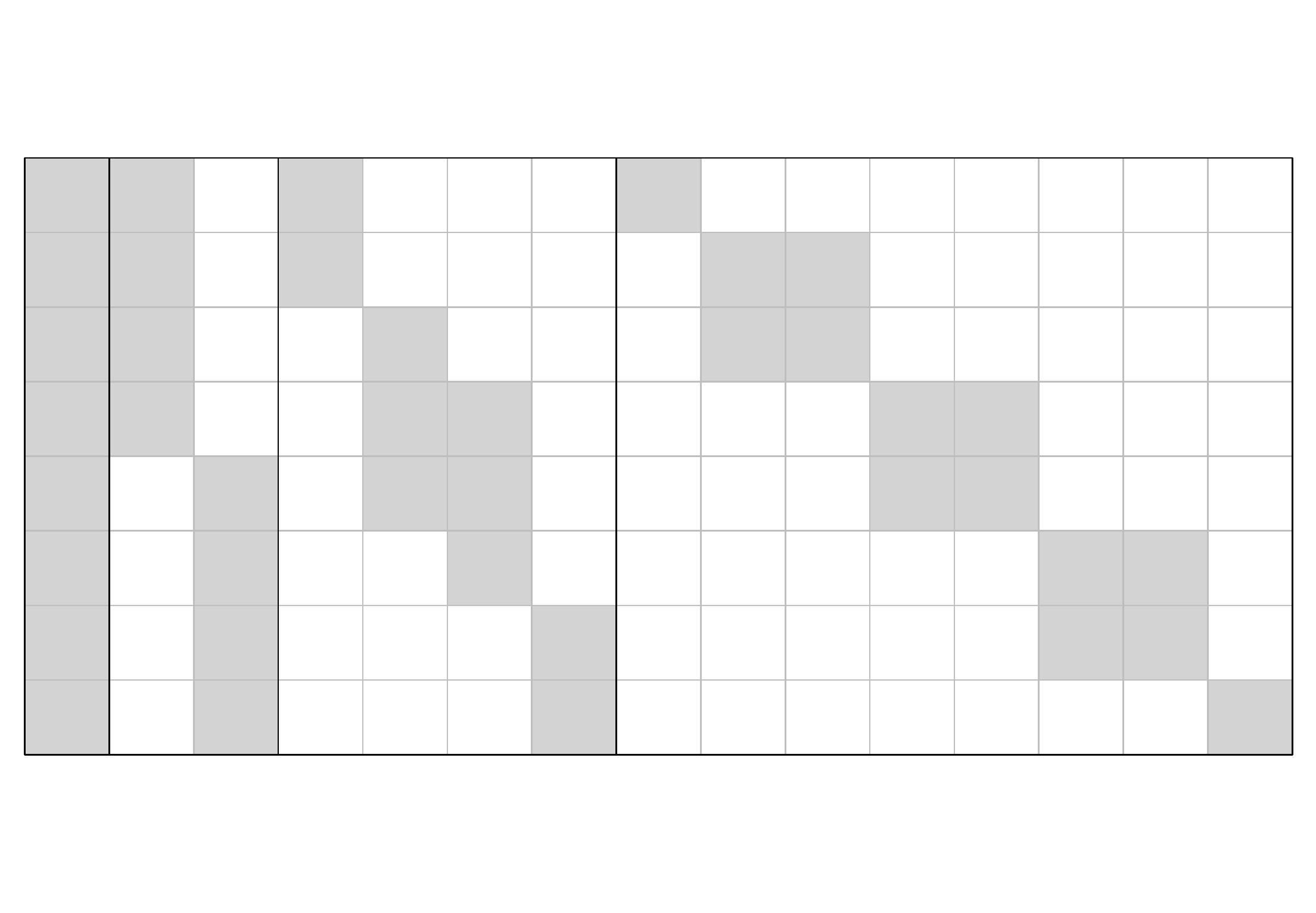} 
		\caption{$\bB$ for $M$-RA-taper}
		\label{fig:TaperB}
	\end{subfigure}%
	\hfill
	\begin{subfigure}{.32\textwidth}
		\centering
		\includegraphics[trim = 20mm 0mm 20mm 0mm, clip,width=.8\linewidth]{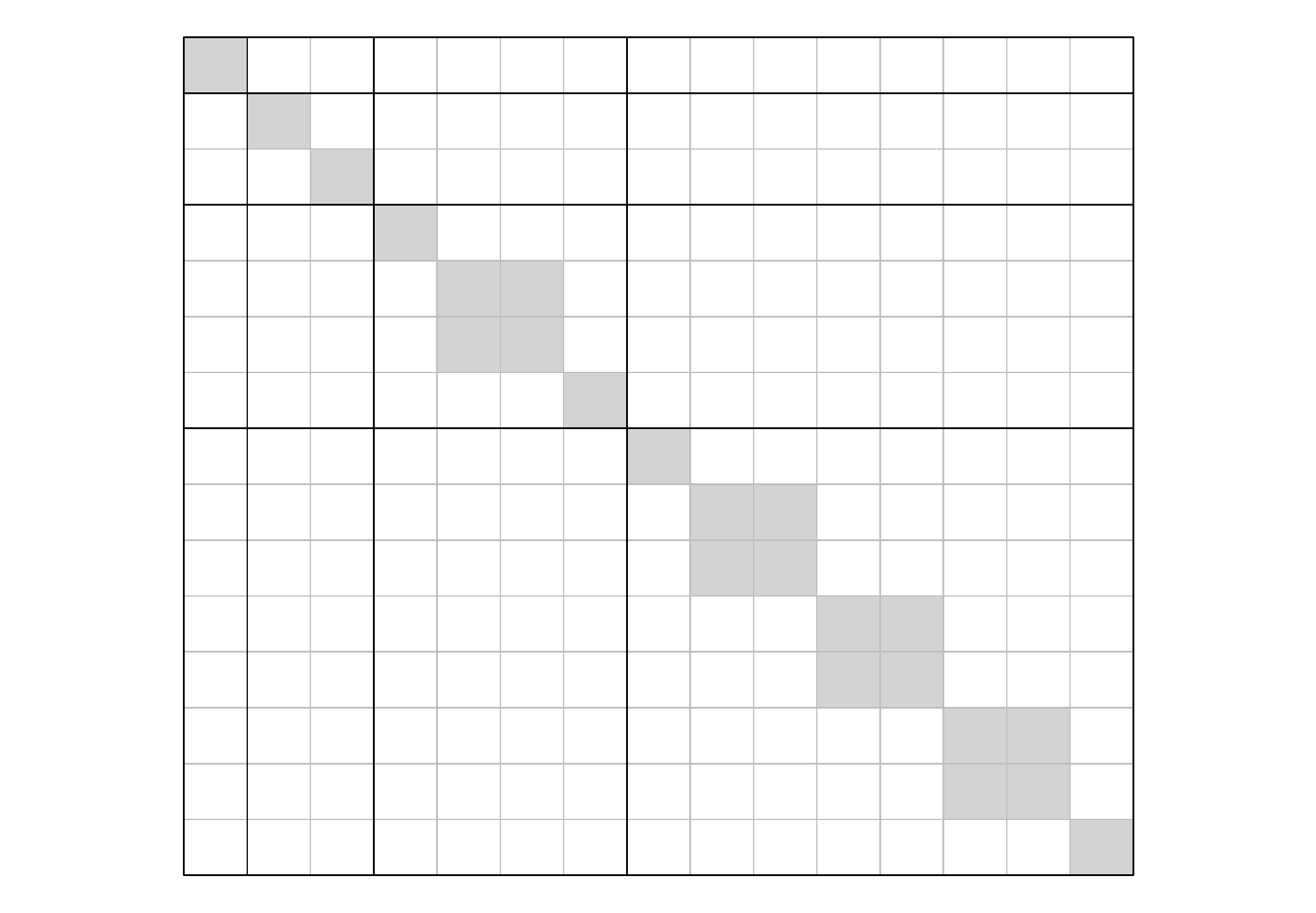} 
		\caption{$\prec$ for $M$-RA-taper}
		\label{fig:TaperP}
	\end{subfigure}
	\hfill
	\begin{subfigure}{.32\textwidth}
		\centering
		\includegraphics[trim = 20mm 0mm 20mm 0mm, clip,width=.8\linewidth]{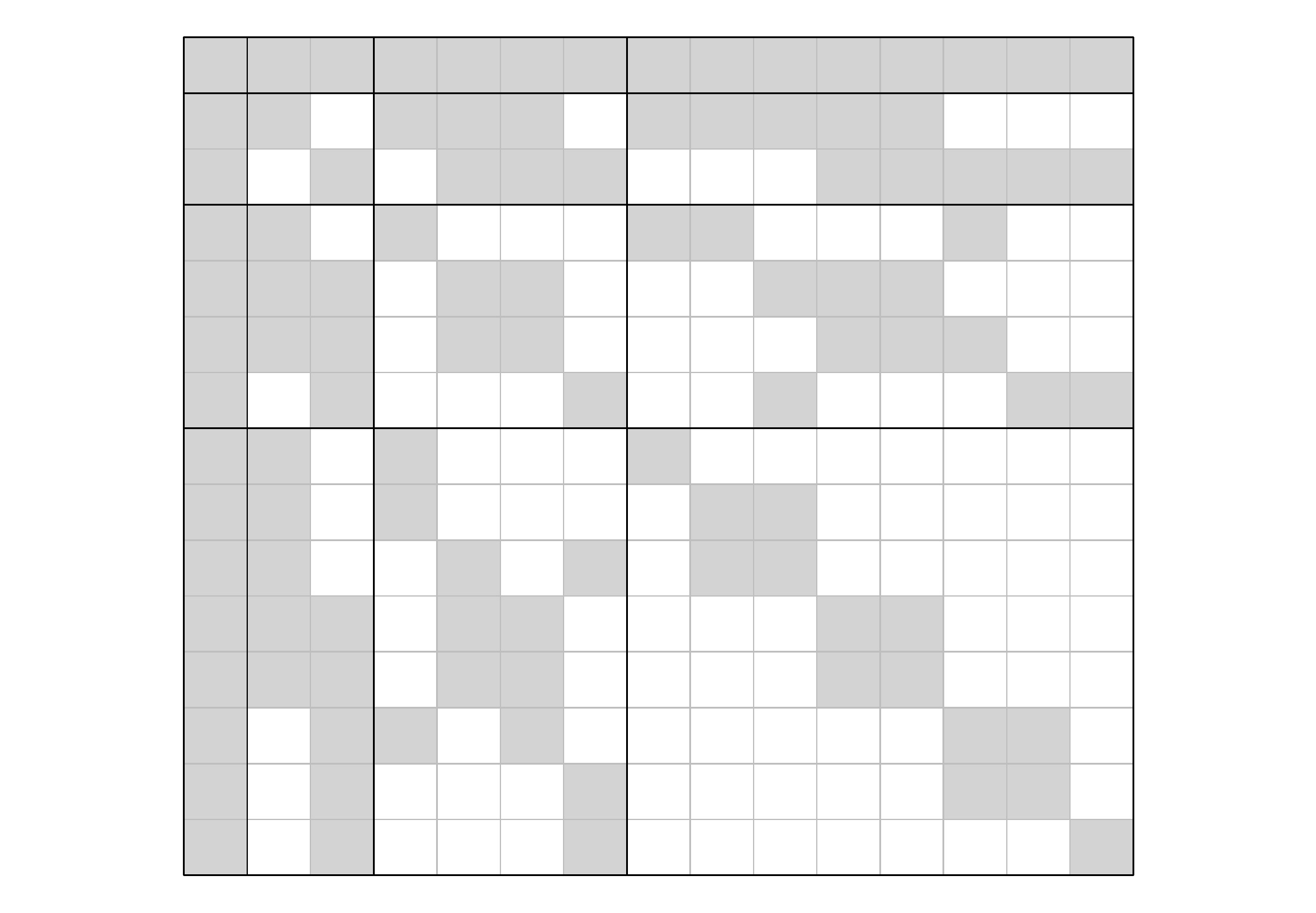}
		\caption{$\pprec$ for $M$-RA-taper}
		\label{fig:TaperPp}
	\end{subfigure}%
	\hfill
	
	\vspace{2mm}
	\begin{subfigure}{.32\textwidth}
		\centering
		\includegraphics[width=.8\linewidth]{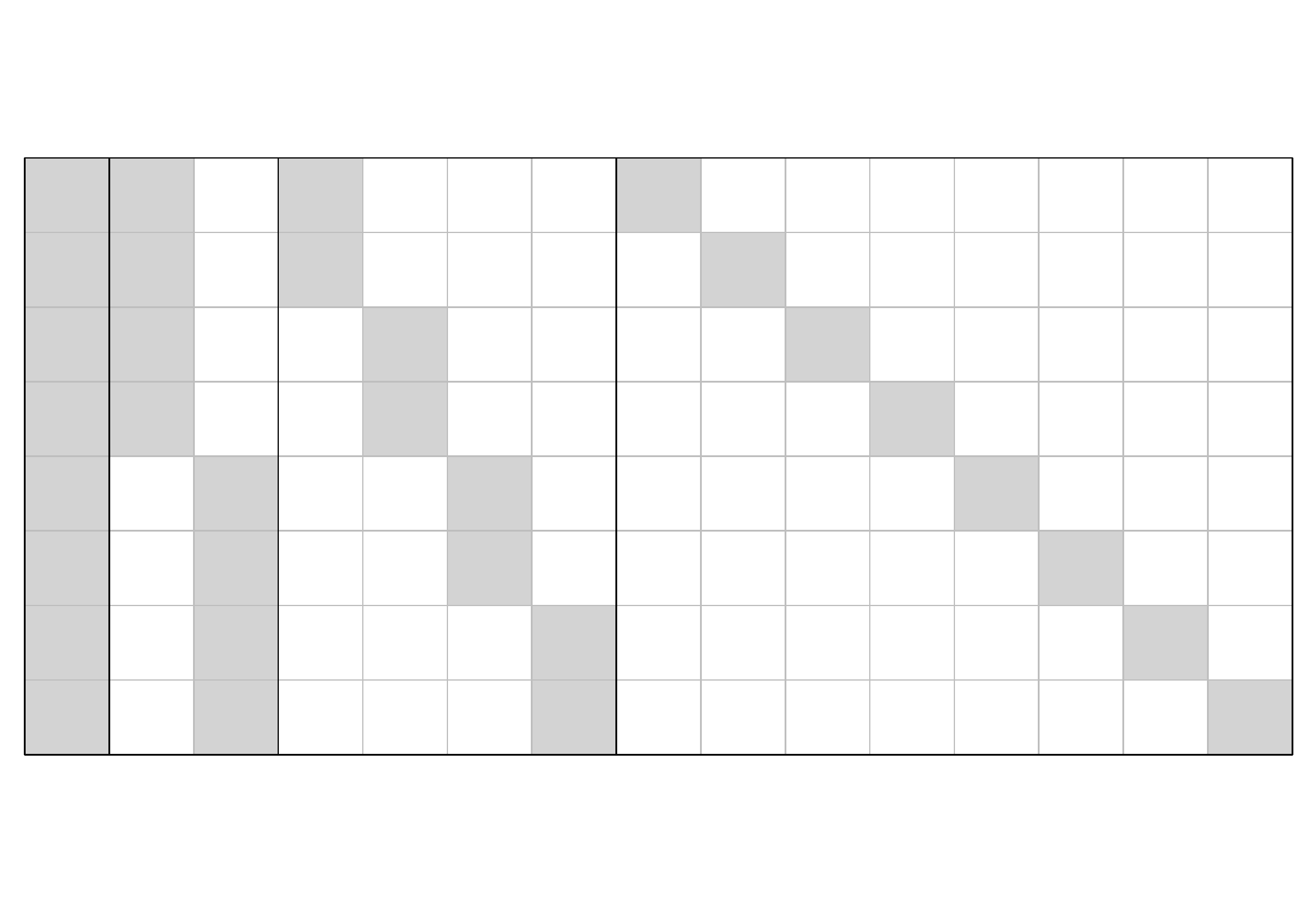}
		\caption{$\bB$ for $M$-RA-block}
		\label{fig:BlockB}
	\end{subfigure}
	\hfill
	\begin{subfigure}{.32\textwidth}
		\centering
		\includegraphics[trim = 20mm 0mm 20mm 0mm, clip,width=.8\linewidth]{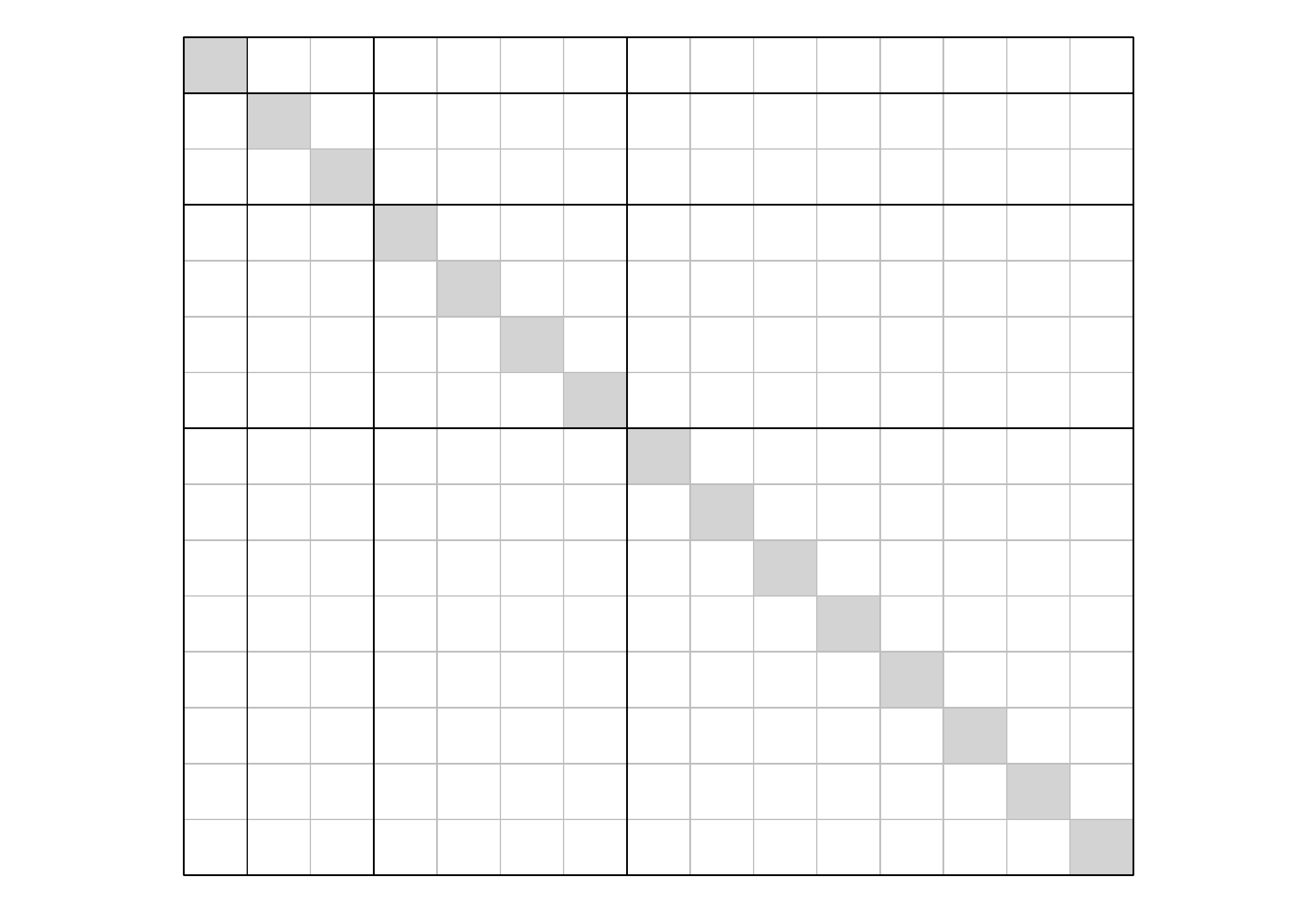}
		\caption{$\prec$ for $M$-RA-block}
		\label{fig:BlockP}
	\end{subfigure}
	\hfill
	\begin{subfigure}{.32\textwidth}
		\centering
		\includegraphics[trim = 20mm 0mm 20mm 0mm, clip,width=.8\linewidth]{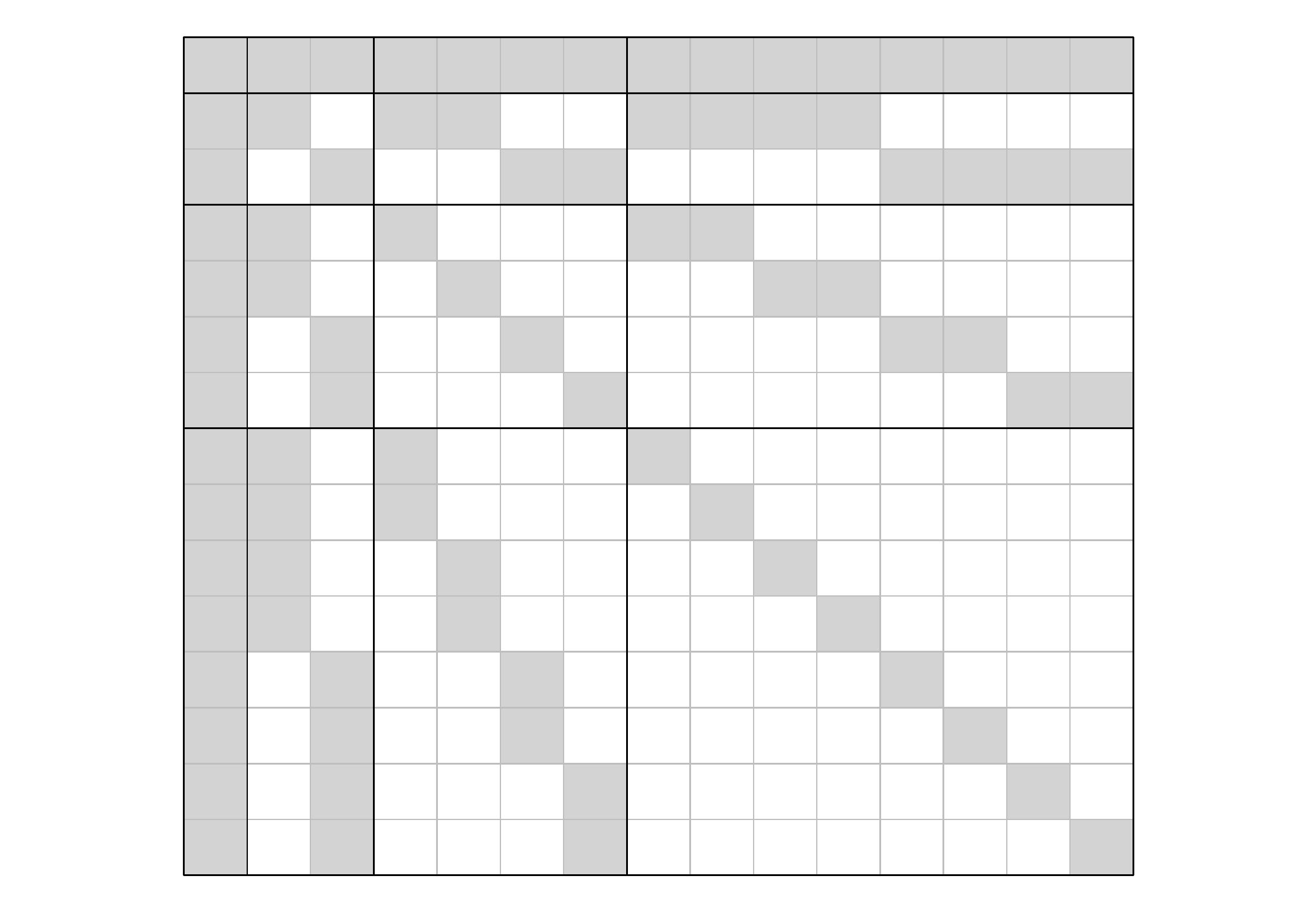}
		\caption{$\pprec$ for $M$-RA-block}
		\label{fig:BlockPp}
	\end{subfigure}
	\caption{Illustration of the sparsity in the matrices $\bB$, $\prec$, and $\pprec$ for the toy example in Figure \ref{fig:toyillus}. Resolutions are separated by solid black lines. Top row: $M$-RA-taper. Bottom row: $M$-RA-block.}
	\label{fig:sparsity}
\end{figure}


\subsubsection{Inference in the absence of measurement error}

If there is no measurement error (i.e., $\bV_\epsilon = \bfzero$), we have
\[
\bz = \by \sim \normal_{n}(\bfzero,\bfsigma).
\]
where $\bfsigma=\bB\prec^{-1}\bB'$. To ensure that $\bB$ (and hence $\bfsigma$) has full rank, we assume for this case that $\locs=\knots$ (and thus $n=r$) and (in light of Proposition \ref{prop:knots}) that the knots are unique.
The likelihood can then be calculated as
$
-2 \log L(\bftheta) = -\log |\bfsigma| -\by'\bfsigma^{-1}\by,
$
where $\log |\bfsigma|= \log|\bB\prec^{-1}\bB'|=\log|\bB|^2-\log|\prec|$, and $\by'\bfsigma^{-1}\by=\widetilde\by'\prec\widetilde\by$ with $\widetilde\by = \bB^{-1}\by$.


\subsection{Inference details for the $M$-RA-block \label{sec:blockinf}}

For the $M$-RA-block from Section \ref{sec:block}, $\bB$, $\prec$, and $\pprec$ are block-sparse matrices, with each block roughly of size $r_0 \times r_0$ and corresponding to (the knots at) a pair of regions. 

As noted in Section \ref{sec:prior}, we have $(\prec_m)_{i,j}=0$ if $\modu_m(\knot_{m,i},\knot_{m,j})=0$, and so $\prec_m$ is a block-diagonal matrix with diagonal blocks $\{v_m(\knots^\jm,\knots^\jm): \jm = 1,\ldots,J\}$, where $\knots^\jm = \{\knot_{m,i}: \knot_{m,i} \in \knots_m \cap \domain_\jm\}$ is the set of roughly $r_0$ knots at resolution $m$ that lie in $\domain_\jm$. It is well known that the inverse $\prec_k^{-1}$ of a block-diagonal matrix $\prec_k$ has the same block-diagonal structure as $\prec_k$, and so the prior calculations in Section \ref{sec:prior} involving $\prec_{k}^{-1}$ can be carried out at low computational cost.

For the posterior covariance matrix, we have from Section \ref{sec:posterior} that $(\pprec_{m,l})_{i,j} = 0$ if $\not\exists \bs \in \domain$ such that $\modu_m(\knot_{m,i},\bs) \neq 0$ and $\modu_l(\knot_{l,j},\bs)\neq 0$, and so the block in $\pprec$ corresponding to regions $\domain_\im$ and $\domain_\jm$ is zero if the regions do not overlap (i.e., if $\domain_\im \cap \domain_\jm = \emptyset$). The Cholesky factor of  a (appropriately reordered) matrix with this particular block-sparse structure has zero fill-in, and can thus be carried out very rapidly. 

\citet{Katzfuss2015} describe an algorithm for inference in a special case of the $M$-RA-block that can be extended to the more general $M$-RA-block considered here. This algorithm is well suited for parallel and distributed computations for massive datasets, and it leads to efficient storage of the full posterior predictive distribution in \eqref{eq:prediction}. The time and memory complexity are shown to be $\order(nM^2r_0^2)$ and $\order(nMr_0)$, respectively.


\subsection{Inference details for the $M$-RA-taper \label{sec:details}}

The case of the $M$-RA-taper from Section \ref{sec:taper} results in sparse matrices, but care must be taken to ensure computational feasibility. A crucial observation for the computational results below is that for any location $\bs \in \domain$ and any resolution $m$, only $\order(r_0)$ knots from $\knots_m$ are within a distance of $d_m$ from $\bs$ (i.e., all sets of the form $\{\knot_{m,i} \in \knots_m: \|\bs - \knot_{m,i}\| \leq d_m\}$ contain only $\order(r_0)$ elements), because we assumed that the $r_m = r_0 J^m$ knots at resolution $m$ are roughly equally spread over the domain $\domain$, and $d_m = d_0 /J^{m/d}$. 

First, consider calculation of the prior matrices as described in Section \ref{sec:prior}. The matrices $\prec$ and $\bB$ have $\order(nr_0)$ and $\order(nMr_0)$ nonzero elements, respectively, because $(\prec_m)_{i,j}=0$ if $\modu_m(\knot_{m,i},\knot_{m,j})=0$, and $(\bB_m)_{i,j} =0$ if $\modu_m(\bs_i,\knot_{m,j})=0$. Before carrying out the actual inference procedures, it is helpful to pre-calculate $\mathcal{I}_{m,l} \colonequals \{(i,j): \modu_l(\knot_{m,i},\knot_{l,j}) \neq 0\}$, the set of nonzero indices of the matrix $\bW^l_{m,l}$, for $l=0,\ldots,m$ and $m=0,\ldots,M$. This can typically be done in $\order(n \log n)$ time \citep[e.g.][]{Vaidya1989}. In the actual inference procedure, we then only need to calculate the $\mathcal{I}_{m,l}$-elements of the matrices $\bW_{m,l}^k$ in \eqref{eq:priorw1}. The main difficulty herein is that while $\prec_k$ is sparse, its inverse $\prec_k^{-1}$ is not. However, we only need to compute certain elements of $\prec_k^{-1}$:
\begin{prop}
\label{prop:priorcomplexity}
For $l=0,\ldots,m$ and $m=0,\ldots,M$, the matrix $\bW_{m,l}^l$ can be obtained by computing 
\begin{equation}
\label{eq:priorwalternate}
\textstyle\bW_{m,l}^{k+1} = \big( \bW_{m,l}^{k} - \bW_{m,k}^{k} \bS_k \bW_{l,k}^{k}{}' \big) \circ \modu_{k+1}(\knots_m,\knots_l), \qquad k=0,\ldots,l-1,
\end{equation}
where $\bS_k = \prec_k^{-1} \circ \bG_k$ and $(\bG_k)_{i,j} = \mathbbm{1}_{\{\| \knot_{m,i} - \knot_{m,j} \| < (2 +2/J) d_m\}}$. Thus, the $(i,j)$ element of $\prec_m^{-1}$ is not required for calculating the prior matrices in \eqref{eq:priorw1} if $\| \knot_{m,i} - \knot_{m,j} \| \geq (2 +2/J) \, d_m$.\\
The total time complexity for computing all prior matrices in \eqref{eq:priorw1} is $\order(n M^2 r_0^3)$, ignoring the cost of computing the $\bS_k$ from the $\prec_k$.
\end{prop}
To calculate $\bS_k$ from $\prec_k$, we use a selected inversion algorithm \citep{Erisman1975,Li2008,Lin2011} in which we regard element $(i,j)$ as a structural zero only if $\| \knot_{k,i} - \knot_{k,j} \| \geq (2 +2/J) d_m$. This algorithm the same computational complexity as the Cholesky decomposition of the same matrix. 
For one-dimensional domains ($d=1$), $\prec_k$ is a banded matrix with bandwidth $\order(r_0)$, and so the time complexity to compute its Cholesky decomposition (and selected inverse) is $\order(r_k r_0^2)$ \citep[e.g.,][p.~187]{Gelfand2010}. For $d \ge 2$, the rows and columns of $\prec$ should be ordered such that the Cholesky decomposition leads to a (near) minimal fill-in and hence fast computations. Functions for this reordering are readily available in most statistical or linear-algebra software. The discussion in \citet{Furrer2006} indicates that the resulting time complexity for the Cholesky decomposition is roughly linear in the matrix dimension for $d=2$. Moreover, our numerical experiments showed that the selected inversions only account for a small fraction of the total time required to compute the prior matrices, and so the total computation time for computing the prior matrices scales roughly as $\order(nM^2r_0^3)$.

Once the prior matrices including $\bB$ and $\prec$ have been obtained, posterior inference requires computing and decomposing the posterior precision matrix $\pprec = \prec + \bB'\bV_\epsilon^{-1}\bB$ in \eqref{eq:postdist}, with $(m,l)$th block $\pprec_{m,l} = \prec_m \mathbbm{1}_{\{m=l\}} + \bB_m'\bV_\epsilon^{-1}\bB_l$. The $(j,k)$th element of this block is
\[
\textstyle(\pprec_{m,l})_{j,k} = (\prec_m)_{j,k} \mathbbm{1}_{\{m=l\}} + \sum_{i=1}^n v_m(\bs_i,\knot_{m,j})v_l(\bs_i,\knot_{l,k}) (\bV_\epsilon)_{i,i}^{-1}.
\]
As each of the $n$ $\bs_i$ is within distances of $d_m$ and $d_l$ of $\order(r_0)$ elements of $\knots_m$ and $\knots_l$, respectively, the time complexity to compute $(\bB'\bB)_{m,l}$ is $\order(nr_0^2)$, and hence computing $\pprec$ requires $\order(nM^2r_0^2)$ time. 
\begin{prop}
\label{prop:posteriornz}
The number of nonzero elements in $\pprec$ is $\order(nMr_0)$.
\end{prop}
The time complexity for obtaining the Cholesky decomposition of $\pprec$ is difficult to quantify, as it depends on its sparsity structure and the chosen ordering, but again our numerical experiments showed that the contribution of the Cholesky decomposition to the overall computation time is relatively small when appropriate reordering algorithms are used.

For prediction, the posterior covariance $\bB^P \pprec^{-1} \bB^P{}'$ in \eqref{eq:prediction} is dense and hence cannot be obtained explicitly for a large number of prediction locations. But the posterior covariance matrix of a moderate number of linear combinations $\bL \by(\locs^P)$ can be obtained as $(\bL \bB^P) \pprec^{-1}(\bL \bB^P)'$, also based on a Cholesky decomposition of $\pprec$.

In summary, the time and memory complexity of the $M$-RA-taper are $\order(nM^2r_0^3)$ and $\order(nMr_0)$, respectively, plus the cost of computing the Cholesky decompositions of $\prec$ and $\pprec$. These decompositions only accounted for a relatively small amount of the overall computation time in our numerical experiments. Thus, the time complexity of the $M$-RA-taper is roughly cubic in $r_0$ while it is square in $r_0$ for the $M$-RA-block. Note that the computational cost for the $M$-RA-taper can be further reduced if the covariance function $C_0$ has a small effective range relative to the size of $\domain$, because then $C_0$ can be tapered at resolution 0 without causing a large approximation error; in contrast, for the $M$-RA-block, we always have $\modu_0(\bs_1,\bs_2) \equiv 1$.
As explained in \citet{Katzfuss2015}, it is often appropriate to expect a good approximation for $M = \mathcal{O}(\log n)$ (and hence $r_0 = \order(1)$), which results in quasilinear complexity as a function of $n$ for the $M$-RA.


\section{Simulation study  \label{sec:simulation}}

For this section, we used data simulated from a true Gaussian process to compare the $M$-RA-block and $M$-RA-taper to full-scale approximations, FSA-block \citep{Sang2011a} and FSA-taper \citep{Sang2012}, which correspond to the $1$-RA-block and $1$-RA-taper, respectively. An implementation of the methods in Julia (http://julialang.org) version 0.4.5 was run on a 16-core machine with 64G RAM.

The true Gaussian process was assumed to have mean zero and an exponential covariance function, 
\begin{equation}
\label{eq:expcov}
 C_0(\bs_1,\bs_2) = \sigma^2 \exp(-\|\bs_1-\bs_2\|/\kappa), \qquad \bs_1,\bs_2 \in \domain,
\end{equation}
with $\sigma^2=0.95$ and $\kappa = 0.05$ on a one-dimensional ($\domain=[0,1]$) or two-dimensional ($\domain=[0,1]^2$) domain. We assumed a nugget or measurement-error variance of $\tau^2 = 0.05$ (i.e., $\bV_\epsilon = 0.05\, \bI$). Results for Mat\'ern covariances with different range, smoothness, and variance parameters showed similar patterns as those presented below and can be found in the Supplementary Material.

\begin{figure}
	\begin{subfigure}{\textwidth}
		\centering
		\includegraphics[trim = 0mm 67mm 15mm 20mm, clip, width=.95\linewidth]{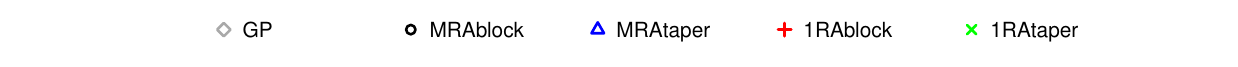}
	\end{subfigure}%
\hfill
	\begin{subfigure}{.5\textwidth}
		\centering
		\includegraphics[width=.95\linewidth]{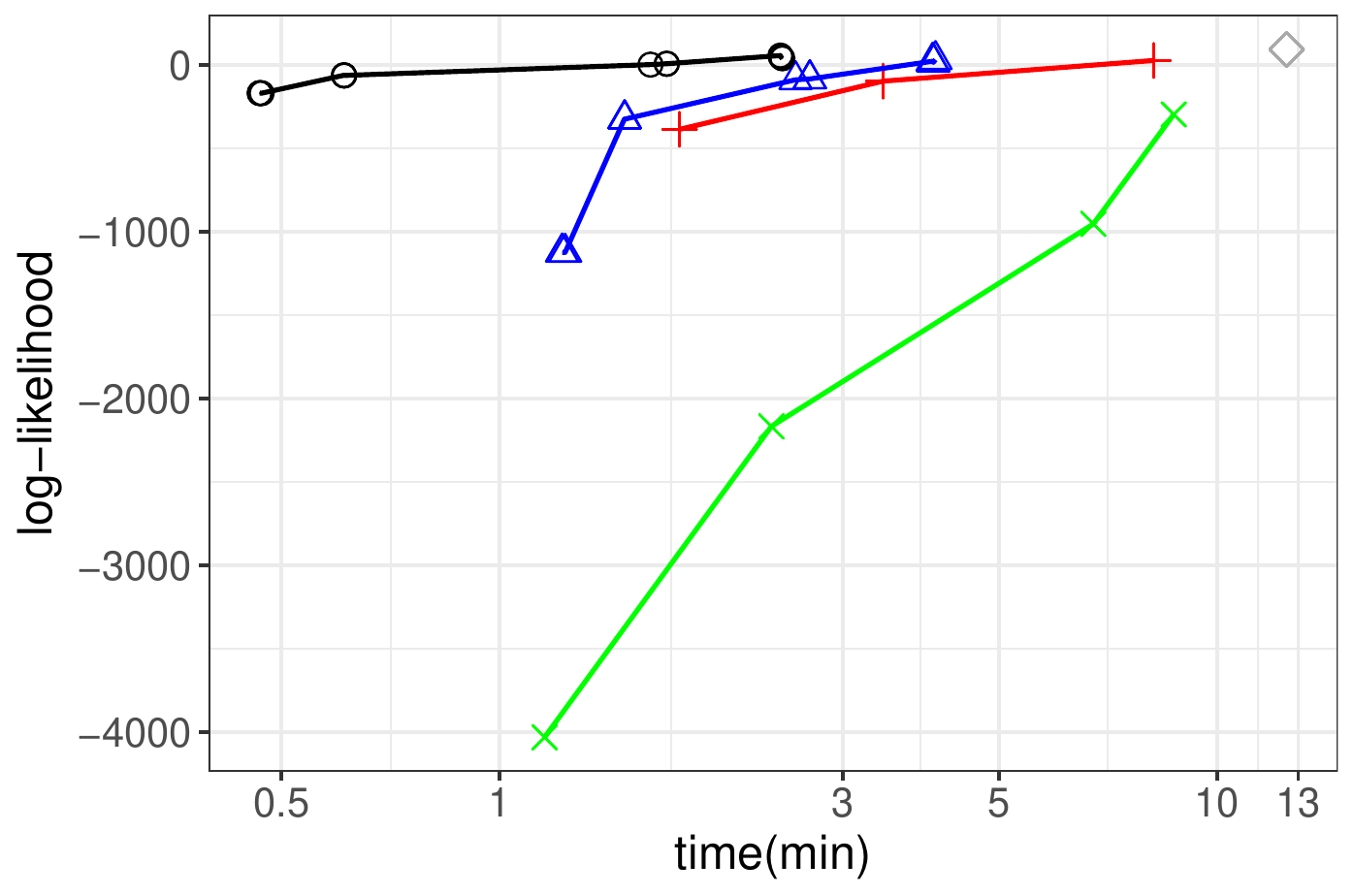}
		\caption{1-D, fixed $n=$ 32,768}
		\label{fig:1DfixN}
	\end{subfigure}%
	\begin{subfigure}{.5\textwidth}
		\centering
		\includegraphics[width=.95\linewidth]{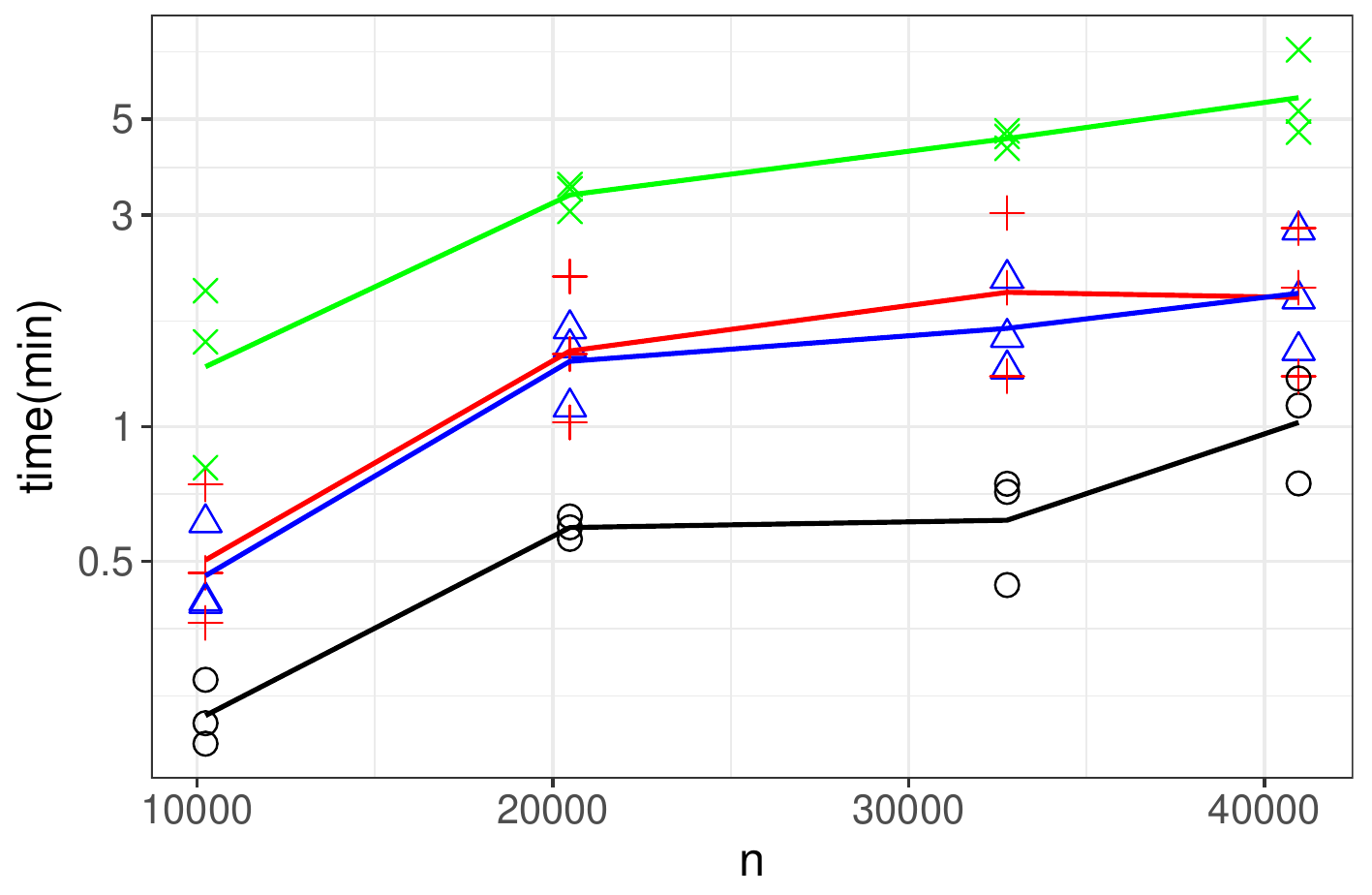}
		\caption{1-D, time for ``close'' approximation}
		\label{fig:1DincN}
	\end{subfigure}

\vspace{3mm}
	\begin{subfigure}{.5\textwidth}
		\centering
		\includegraphics[width=.95\linewidth]{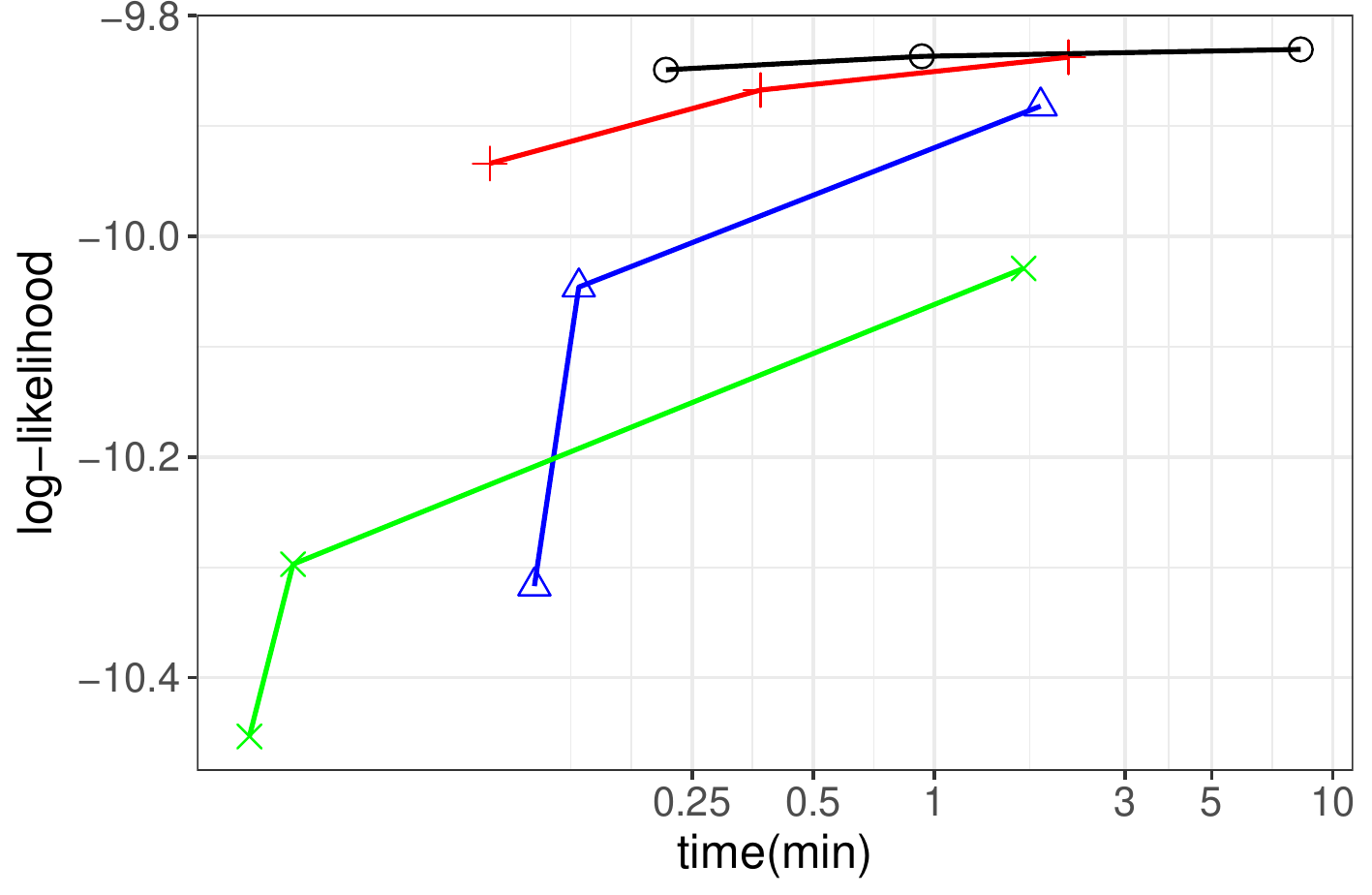}
		\caption{2-D, fixed $n=$ 36,864}
		\label{fig:2DfixN}
	\end{subfigure}%
	\begin{subfigure}{.5\textwidth}
		\centering
		\includegraphics[width=.95\linewidth]{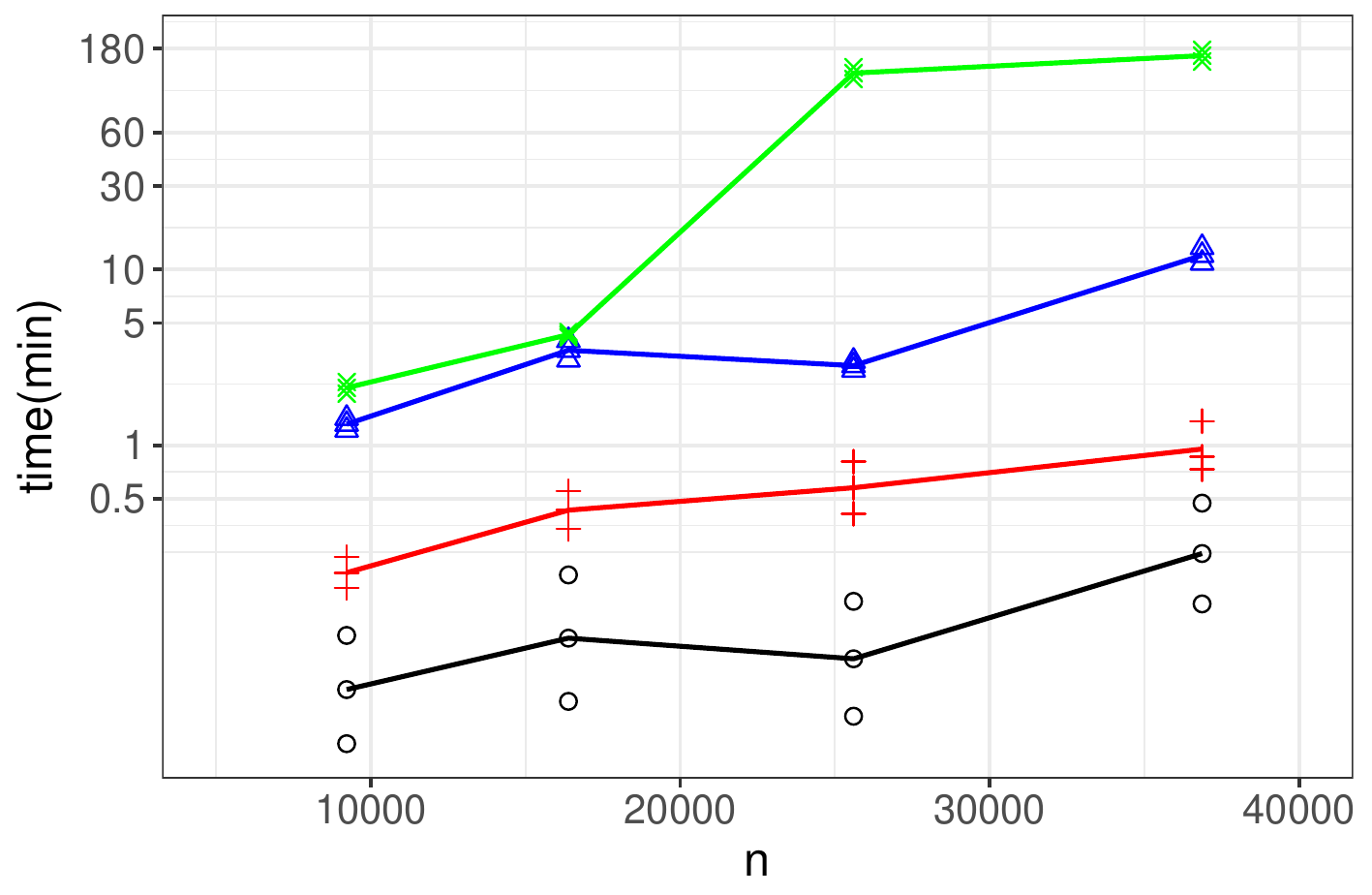}
		\caption{2-D, time for ``close'' approximation}
		\label{fig:2DincN}
	\end{subfigure}
\caption{Summary of results from the simulation study. Top row: $\domain=[0,1]$. Bottom row: $\domain=[0,1]^2$. Left column: Log-score versus computation time for different versions of the $M$-RA for fixed $n$. Right column: Computation time required to get a ``close'' approximation to the truth (or best approximation) for different $n$; lines connect the means of the three times for each model and each $n$. 
Note that all time axes are on a log scale.
Additional results can be found in the Supplementary Material.}
\label{fig:simstudy}
\end{figure}

All comparisons were carried out based on the log-score (i.e., the log-likelihood at the true parameter values), which is a strictly proper scoring rule that is uniquely maximized in expectation by the true model \citep[e.g.,][]{Gneiting2014}. All results were averaged over five replications.

For $M$-RA-taper, some experimentation showed that there are general guidelines to follow in order to get a close approximation to true GP. For a true covariance function $C_0$ with effective range $\rho$, we recommend setting the $M$-RA taper range at resolution $0$ to $d_0 = 2\rho$, and the distance between two adjacent knots at resolution 0 to be at most $\frac{2}{3}$ of $\rho$. For example, the covariance in \eqref{eq:expcov} has an effective range of $\rho \approx 0.15$, and so we set $d_0=0.3$ and the distance between adjacent knots at resolution 0 to $0.1$.

First, we simulated datasets of different sizes on an equidistant grid in one dimension with $\domain=[0,1]$, which permitted fast simulation using the Davies-Harte algorithm and evaluation of the exact likelihood using the Durbin-Levinson algorithm for comparison \citep{McLeod2007}. For each dataset, we recorded the computation times and log-scores for different versions of the $M$-RA (i.e., with different $r_0$, $J$, and $M$). We also considered the computation times to achieve particular levels of approximation accuracy, specifically the time required to obtain an average log-score within a difference of $0.003n$, $0.005n$, and $0.007n$ of the log-score of the true model.
We then repeated the simulation study in two dimensions, $\domain=[0,1]^2$. As it was infeasible to compute the true log-likelihood for large $n$, we use the best approximation (i.e., the largest approximated log-likelihood) as the base to compare the relative performance of different methods, with cut-off values of $0.008n$, $0.01n$, and $0.012n$.

The results are summarized in Figure \ref{fig:simstudy}. The computation times scaled roughly as expected. The $M$-RA-block was consistently better than the other methods, while $M$-RA-taper and 1-RA-block performed similarly. The 1-RA-taper was not competitive.


\section{Application \label{sec:application}}

In this section, we applied the four methods from Section \ref{sec:simulation} to a real satellite dataset. We considered $n=44{,}711$ Level-3 daytime sea surface temperature (SST) data from August 2016 over a region in the North Atlantic Ocean, as measured by the Moderate Resolution Imaging Spectroradiometer on board the Terra satellite. The data are freely available at \url{https://giovanni.gsfc.nasa.gov}.
More specifically, the data (shown in Figure \ref{fig:res}) were taken to be the residuals of the SST data after removing a longitudinal and latitudinal trend. The exploratory analysis showed that an exponential covariance fit the data well, and so all methods used were approximating the covariance in \eqref{eq:expcov}. We assumed a constant noise variance $\tau^2$ (i.e., $\bV_\epsilon = \tau^2 \bI$).

\begin{figure}
\captionsetup{justification=centering}
	\begin{subfigure}{.33\textwidth}
		\centering
		\includegraphics[width =.98\linewidth]{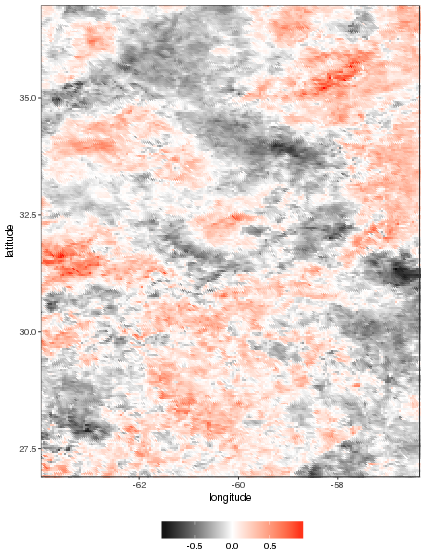}
		\caption{Complete set of SST residuals after removing trend}
		\label{fig:res}
	\end{subfigure}%
	\hfill
	\begin{subfigure}{.33\textwidth}
		\centering
		\includegraphics[width =.98\linewidth]{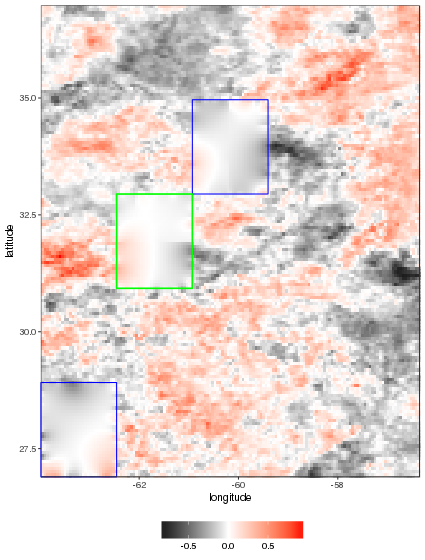}
		\caption{$M$-RA-block with $r_0=121$, $J=4$, $M=4$}
		\label{fig:block1set2}
	\end{subfigure}%
	\hfill
	\begin{subfigure}{.33\textwidth}
		\centering
		\includegraphics[width =.98\linewidth]{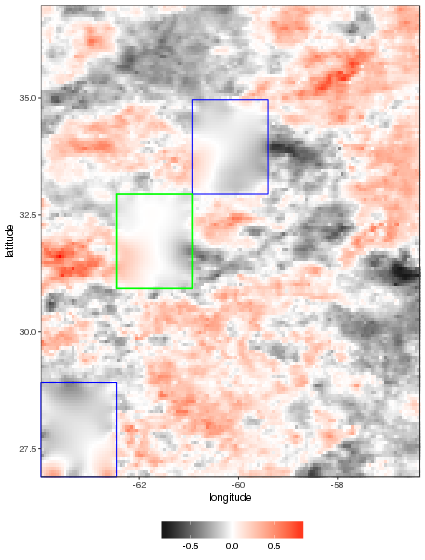}
		\caption{$M$-RA-taper with $r_0=576$, $J=4$, $M=3$}
		\label{fig:taper1set2}
	\end{subfigure}%
	\hfill
	\begin{subfigure}{.33\textwidth}
		\centering
		\phantom{\includegraphics[width =.98\linewidth]{figures/res.png}}
	\end{subfigure}%
	\hfill
	\begin{subfigure}{.33\textwidth}
		\centering
		\includegraphics[width =.98\linewidth]{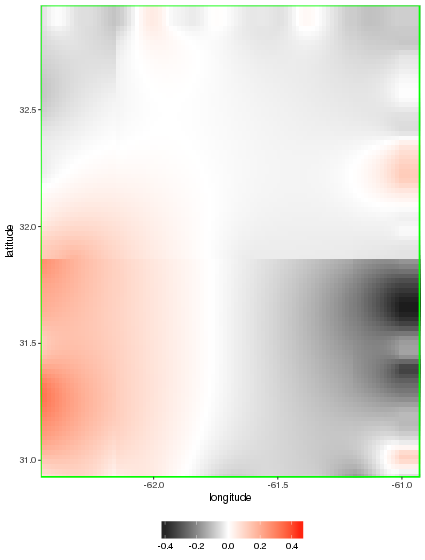}
		\caption{Missing-area prediction for $M$-RA-block}
		\label{fig:zoominblock1}
	\end{subfigure}%
	\hfill
	\begin{subfigure}{.33\textwidth}
		\centering
		\includegraphics[width =.98\linewidth]{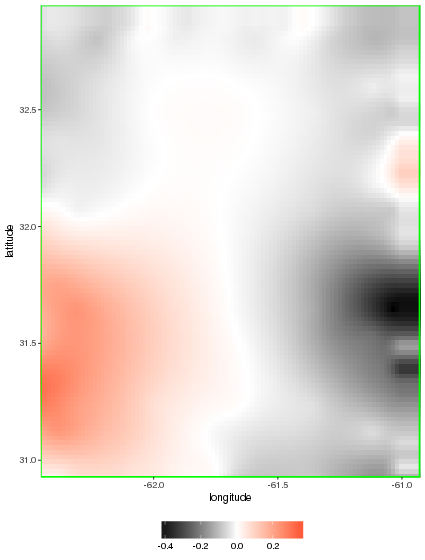}
		\caption{Missing-area prediction for $M$-RA-taper}
		\label{fig:zoomintaper1}
	\end{subfigure}%
\captionsetup{justification=justified}
\caption{Top row: Complete dataset of sea-surface temperature, along with posterior predictive means for $M$-RA-taper and $M$-RA-block based on removing three areal test regions and additional randomly selected values. Bottom row: Zoomed-in view of the green rectangle in the upper prediction plots. Color scales are in units of degrees Celsius.}
\label{fig:prediction1}
\end{figure}

To compare the different approximation methods, we created five different datasets by randomly splitting the complete dataset of residuals into training data, areal test data, and random test data, each containing 78\%, 12\% and 10\%, respectively, of the values in the full data set. The split of the complete data into training and test sets was designed to mimic the typical setting of Level-2 satellite data, with unobserved areas over which the satellite did not fly in a particular time period, and observed areas with some missing values (e.g., due to clouds). Specifically, the areal test locations were obtained by splitting the domain into $5 \times 5 =25$ equal-area rectangles and then removing three of these rectangles at random. The remaining test locations were obtained by simple random sampling of the remaining locations.

Based on each of the five training sets and for a range of settings for each of the four approximation methods, we carried out maximum-likelihood estimation of the unknown parameters $\sigma^2$, $\kappa$, and $\tau^2$, and obtained posterior predictive distributions at the held-out test locations. 
We compared the pointwise (i.e., marginal) posterior distributions obtained by the methods to the held-out test data in terms of the root mean squared prediction error (RMSPE) and the continuous ranked probability score (CRPS), which is a strictly proper scoring rule that quantifies the fit of the entire predictive distribution to the data \citep[e.g.,][]{Gneiting2014}. The scores for the random test data were almost zero for all methods. The scores for the areal test data are shown in Figure \ref{fig:rmspecrps} (averaged over the five datasets). In general, the scores for $M$-RA-taper and $M$-RA-block were better than those for the full-scale approximations. $M$-RA-taper produced some RMSPEs that were even lower than those for $M$-RA-block.

\begin{figure}
	\begin{subfigure}{\textwidth}
		\centering
		\includegraphics[trim = 0mm 85mm 50mm 0mm, clip, width=.7\linewidth]{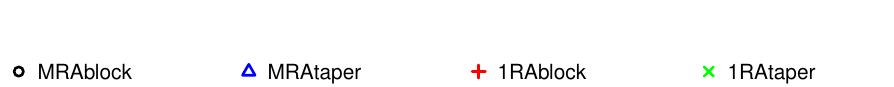}
	\end{subfigure}%
\hfill
	\begin{subfigure}{0.5\textwidth}
	\centering 
	\includegraphics[width =.98\linewidth]{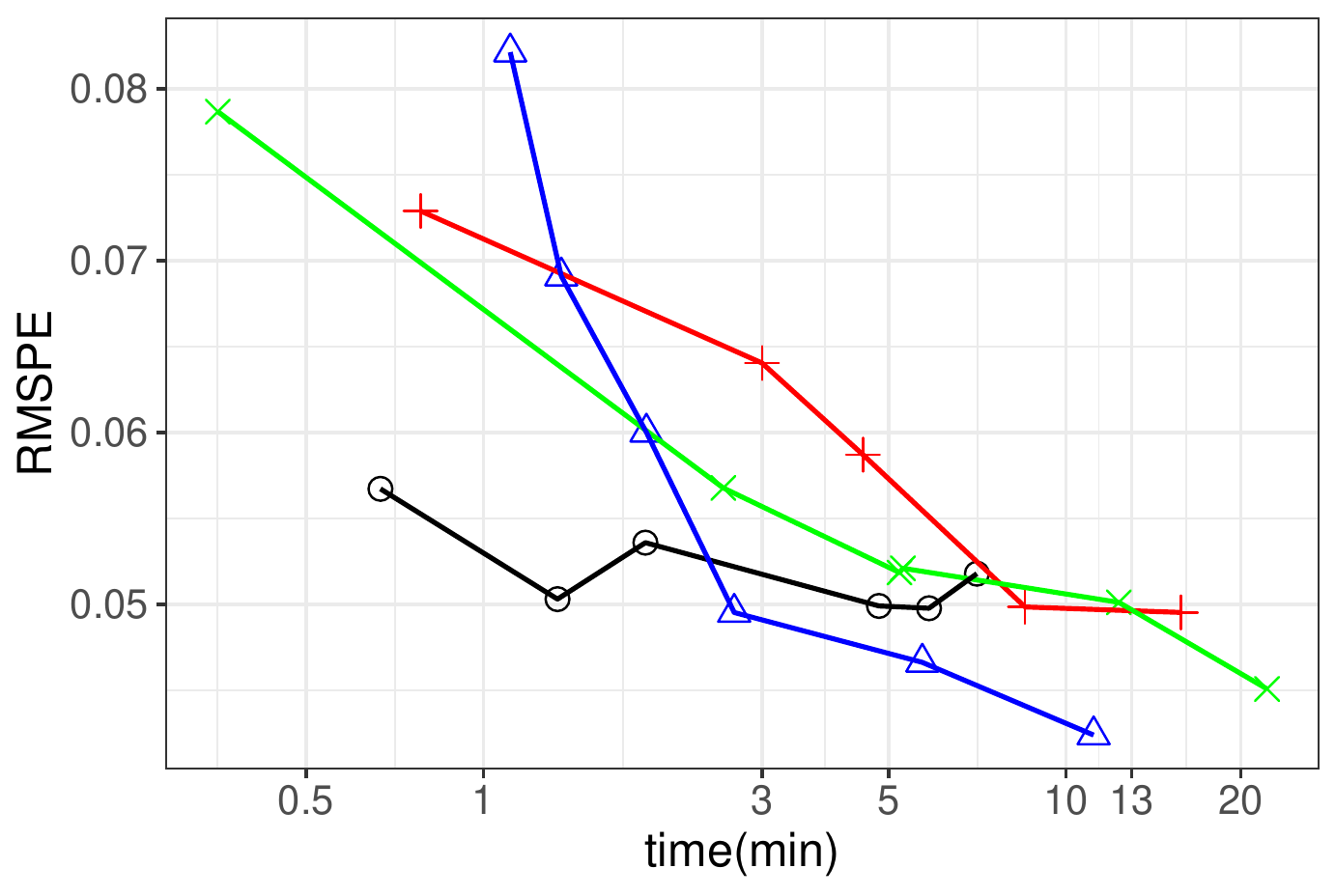}
	\caption{RMSPE}
	\label{fig:sstrmspe}
	\end{subfigure}%
\hfill
	\begin{subfigure}{0.5\textwidth}
		\centering
		\includegraphics[width =.98\linewidth]{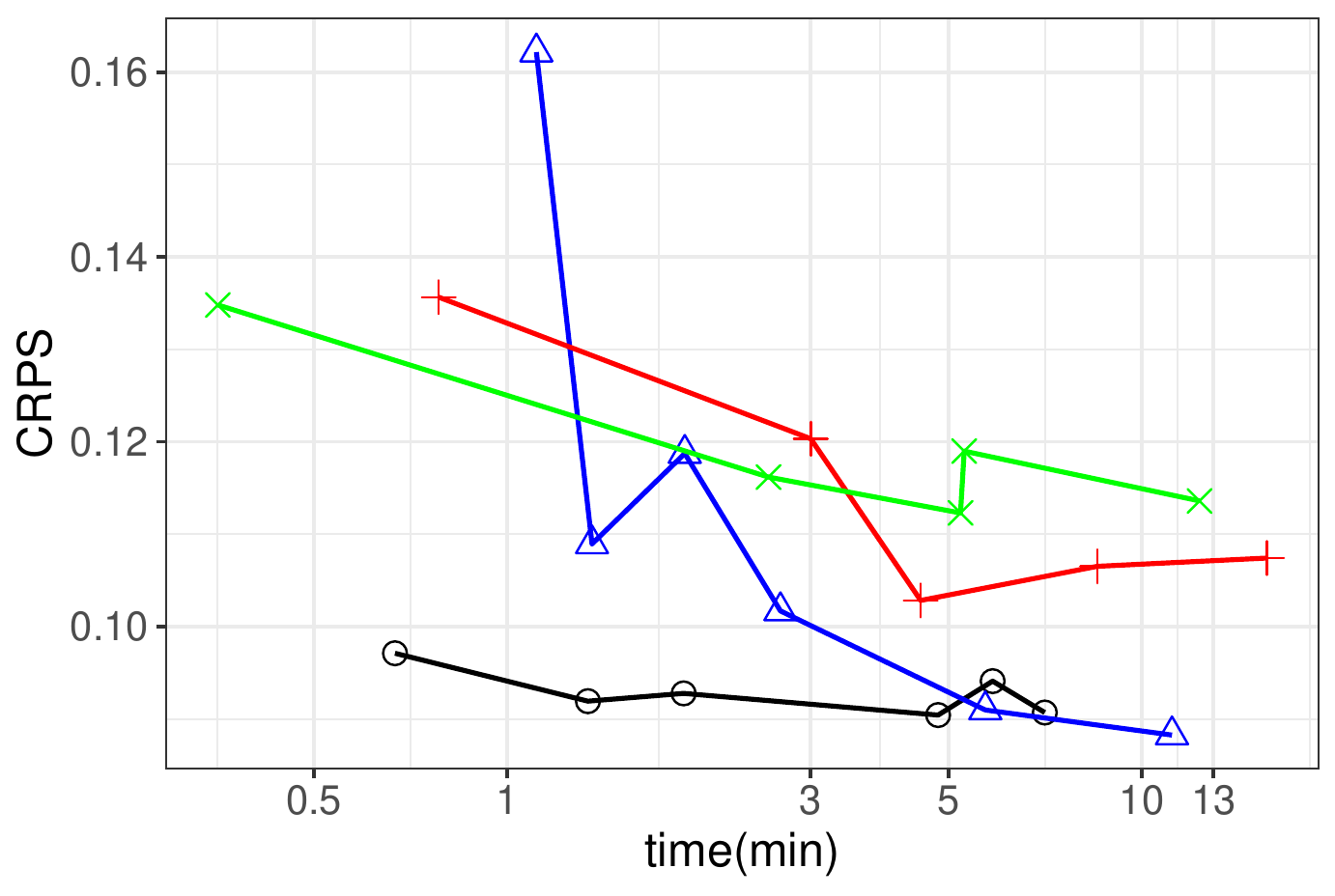}
		\caption{CRPS}		
	\end{subfigure}%
\caption{For the satellite SST data, comparison of scores (lower is better) for predictions of areal test data for different settings of the $M$-RA}
\label{fig:rmspecrps}
\end{figure}

Maybe more important than the differences in prediction scores are the differences in the prediction plots. Figure \ref{fig:prediction1} shows an example of the posterior means as obtained by $M$-RA-taper and $M$-RA-block, for versions of the two methods that took a similar time to run (5 to 7 minutes) and resulted in similar RMSPEs in Figure \ref{fig:sstrmspe}. Despite the good approximation accuracy and low RMSPE of $M$-RA-block, we can see in Figure \ref{fig:zoominblock1} that there are clearly visible artifacts due to discontinuities of the $M$-RA-block at the region boundaries (see Proposition \ref{prop:continuity}), which do not appear for the continuous $M$-RA-taper in Figure \ref{fig:zoomintaper1}. Avoiding these kinds of ``non-physical'' artifacts is often of paramount importance to domain scientists.


\section{Conclusions \label{sec:conclusions}}
	
We have proposed and studied a general approach for obtaining multi-resolution approximations of Gaussian processes (GPs) based on an orthogonal decomposition of the GP of interest into processes at multiple resolutions. We considered two specific cases of this approach: The $M$-RA-taper achieves sparsity and computational feasibility by applying increasingly compact isotropic tapering functions as the resolution increases, while the $M$-RA-block is based on a recursive block-partitioning of the spatial domain and assumes conditional independence between the spatial subregions at each resolution. We have provided algorithms for inference, along with computational complexity of the methods. Within our framework, one could also consider other partitioning schemes or nonstationary tapering, which might be especially useful when approximating nonstationary processes.

We have shown theoretically and numerically that both $M$-RA versions have useful properties and can outperform related existing approaches. The $M$-RA-block achieves more accurate approximations to a given covariance function for a given computation time, and its block-sparse structure allows it to approximate the likelihood for truly massive datasets on modern distributed computing systems. However, the $M$-RA-block process is discontinuous at the subregion boundaries, which can be undesirable in prediction maps.
The $M$-RA-taper can be useful for real-world applications in which the true covariance function is unknown anyway, and hence it might be more important to have a ``smooth'' model that avoids the potential artifacts and discontinuities inherent to the $M$-RA-block due its domain partitioning. The $M$-RA-taper's prediction accuracy can be highly competitive, especially when the effective range of the covariance model is small relative to the domain size. Also note that posterior inference involving the $M$-RA-taper only requires general sparse matrices, which would allow for relatively straightforward treatment of areal-averaged measurements (e.g., satellite footprints). 
	
Future work will consider multivariate, spatio-temporal, and non-Gaussian extensions of the methodology. Also of interest is more precise quantification of the approximation error, and a further investigation of how to choose the number of resolutions and the knots depending on the covariance to be approximated. While our methods are, in principle, also applicable in the context of GP regression, some additional consideration of the choice of knots and partitions in high-dimensional covariate spaces would be warranted.

\footnotesize
\appendix
\section*{Acknowledgments}

The authors were partially supported by National Science Foundation (NSF) Grant DMS--1521676. Katzfuss' research was also partially supported by NSF CAREER Grant DMS--1654083. We would like to thank Tim Davis, Shiyuan He, and several anonymous reviewers for helpful comments. The data used in the application were acquired as part of the mission of NASA's Earth Science Division and archived and distributed by the Goddard Earth Sciences (GES) Data and Information Services Center (DISC).

\section{Proofs \label{app:proofs}}

In this section, we provide proofs for the propositions stated throughout the manuscript. We also state and prove three lemmas that are used in the proofs of the propositions.

\begin{proof}[Proof of Proposition \ref{prop:mracov}]
From \eqref{eq:stacked}, we have $y_M(\cdot) = \bb(\cdot)'\bfeta$, where $\bfeta \sim \normal_{r}(\bfzero,\prec^{-1})$ and $\bb(\cdot)$ is a vector of deterministic functions (for given $C_0$, $\knots$, and $\modu$). Hence, it is trivial to show that $y_M(\cdot)$ is a Gaussian process with mean zero. The covariance function is derived by combining the expression for $y_M(\cdot)$ on the right-hand side of \eqref{eq:mra1} with the equations in \eqref{eq:mradetails}.
\end{proof}

\begin{lemma}[Exact predictive process]
\label{lem:pp}
The predictive process is exact at any knot location; that is, if $x^{(m)}(\cdot)$ is the predictive process of $x(\cdot) \sim \GP(0,C)$ based on knots $\knots_m$ (see Definition \ref{def:pp}), and $\bs_1 \in \knots_m$ (or $\bs_2 \in \knots_m$), then 
\[
\cov\big(x^{(m)}(\bs_1),x^{(m)}(\bs_2)\big) = C(\bs_1,\bs_2).
\]
\end{lemma}
\begin{proof}[Proof of Lemma \ref{lem:pp}]
By the law of total covariance, we have
\begin{align*}
\cov\big(x^{(m)}(\bs_1),x^{(m)}(\bs_2)\big) & = \cov\big(\E(x(\bs_1)|\bx(\knots_m), \E(x(\bs_2)|\bx(\knots_m) \big) \\
& = \cov\big(x(\bs_1),x(\bs_2)\big) -\E\big(\cov(x(\bs_1),x(\bs_2)|\bx(\knots_m) ) \big) = C(\bs_1,\bs_2),
\end{align*}
because $\cov(x(\bs_1),x(\bs_2)|\bx(\knots_m) )=0$ if $\bs_1 \in \knots_m$ (or $\bs_2 \in \knots_m$).
\end{proof}

\begin{proof}[Proof of Proposition \ref{prop:knots}]
The proof will be carried out by induction. For $l=1$, we have $v_{m+1}(\knot,\bs) = \big(v_m(\knot,\bs) - \cov(\ppm(\knot),\ppm(\bs))\big) \modu_{m+1}(\knot,\bs)=0$, because using Lemma \ref{lem:pp}, we can see that $\cov\big(\ppm(\knot),\ppm(\bs)\big)=\cov\big(\deltam^{(m)}(\knot),\deltam^{(m)}(\bs)\big) = \cov\big(\deltam(\knot),\deltam(\bs)\big) = v_m(\knot,\bs)$. For $l>1$, assuming that $v_{m+l-1}(\knot,\bs) =0$, we have 
\[
v_{m+l}(\knot,\bs)= \big(v_{m+l-1}(\knot,\bs) - \bb_{m+l-1}(\knot)'\prec_{m+l-1}^{-1} \bb_{m+l-1}(\bs)\big)\cdot\modu_{m+l}(\knot,\bs)=0,
\]
because $\bb_{m+l-1}(\knot)=v_{m+l-1}(\knot,\knots_{m+l-1})=\bfzero$.
\end{proof}

\begin{lemma}[$M$-RA covariance at knot location $\bs$]
\label{lem:covknot}
If $\bs_1 \in \knots$, then 
\[
\textstyle C_M(\bs_1,\bs_2) = \sum_{m=0}^{M-1} v_m(\bs_1,\knots_m)v_m(\knots_m,\knots_m)^{-1}v_m(\knots_m,\bs_2) + v_M(\bs_1,\bs_2), \quad \bs_2 \in \domain.
\]
\end{lemma}
\begin{proof}[Proof of Lemma \ref{lem:covknot}]
In the expression for $C_M$ in Proposition \ref{prop:mracov}, we have $v_M(\bs_1,\knots_M)v_M(\knots_M,\knots_M)^{-1}v_M(\knots_M,\bs_2) = v_M(\bs_1,\bs_2)$ for $\bs_1 \in \knots$. This follows from Lemma \ref{lem:pp} if $\bs_1 \in \knots_M$, and from Proposition \ref{prop:knots} for $\bs_1 \in \knots_m$ for $m<M$ (because then both sides of the equation are zero).
\end{proof}

\begin{proof}[Proof of Proposition \ref{prop:variance}]
Because $\modu_m(\bs,\bs)=1$ for all $m=0,1,\ldots,M$, we have $v_0(\bs,\bs)=C_0(\bs,\bs)$ and $v_{m+1}(\bs,\bs) =v_{m}(\bs,\bs) - \bb_m(\bs)'\prec_m^{-1} \bb_m(\bs)$ for $m=1,\ldots,M$. Thus, we can write $v_M(\bs,\bs) = C_0(\bs,\bs) - \sum_{m=0}^{M-1} \bb_m(\bs)'\prec_m^{-1}\bb_m(\bs)$, and using Lemma \ref{lem:covknot}, we have
\[
\textstyle C_M(\bs,\bs) = \sum_{m=0}^{M-1} \bb_m(\bs)'\prec_m^{-1}\bb_m(\bs) + C_0(\bs,\bs) - \sum_{m=0}^{M-1} \bb_m(\bs)'\prec_m^{-1}\bb_m(\bs) = C_0(\bs,\bs).
\]
\end{proof}


\begin{proof}[Proof of Proposition \ref{prop:smoothness}]
First, note that $y_0(\cdot)$ is $p$ times (mean-square) differentiable at $\bs$ if and only if $C_{0,\bs}(\bh) \colonequals C_0(\bs,\bs+\bh)$ is $2p$ times differentiable at the origin ($2p$DO).

By Lemma \ref{lem:covknot}, we have $C_{M,\bs}(\bh) \colonequals C_M(\bs,\bs+\bh)=\sum_{m=0}^{M-1}f_m(\bs,\bs+\bh) + v_M(\bs,\bs+\bh)$, where $f_m(\bs_1,\bs_2) \colonequals \sum_{j=1}^{r_m}a_{m,j}(\bs_1)v_m(\knot_{m,j},\bs_2)$, and $a_{m,j}(\bs)$ is the $j$-th element of the vector $\ba_m(\bs)=v_m(\knots_m,\knots_m)^{-1}v_m(\knots_m,\bs)$. We now show by induction for $m=0,\ldots,M-1$ that
\begin{idea}
\label{eq:statement}
  $v_{m,\knot,\bs}(\bh) \colonequals v_m(\knot,\bs+\bh)$ (for any $\knot \in \knots$) and $f_{m,\bs}(\bh) \colonequals f_m(\bs,\bs+\bh)$ are at least $2p$DO, and $v_{m,\bs,\bs}(\bh)$ is exactly $2p$DO.
\end{idea}
For $m=0$, $v_{0,\knot,\bs}(\bh) =C_0(\knot,\bs+\bh) \cdot \modu_0(\knot,\bs+\bh)$ is at least $2p$DO by assumption and hence so is $f_{0,\bs}(\bh)=\sum_{j=1}^{r_0}a_{0,j}(\bs)v_0(\knot_{0,j},\bs+\bh)$. Further, $v_{0,\bs,\bs}(\bh)$ is exactly $2p$DO. Now assume that (\ref{eq:statement}) holds for $m$. Then, using Equation \ref{eq:decompdetails}, $v_{m+1,\knot,\bs}(\bh)=\big( v_{m,q,\bs}(\bh) - f_m(\knot,\bs+\bh) ) \cdot \modu_{m+1}(\knot,\bs+\bh)$, which is at least $2p$DO, and so is $f_{m+1,\bs}(\bh)=\sum_{j=1}^{r_{m+1}}a_{m+1,j}(\bs)v_{m+1,\knot_{m,j},\bs}(\bh)$. Also, $v_{m+1,\bs,\bs}(\bh)$ is exactly $2p$DO. This proves (\ref{eq:statement}) for $m=1,\ldots,M$.

In summary, we have $C_{M,\bs}(\bh)= \sum_{m=0}^{M-1}f_{m,\bs}(\bh) + (v_{M-1,\bs,\bs}(\bh)-f_{M-1,\bs}(\bh))\cdot \modu_M(\bs,\bs+\bh)$, where $\modu_{M,\bs}(\bh)=\modu_M(\bs,\bs+\bh)$ and $f_{m,\bs}(\bh)$, $m=0,\ldots,M-1$, are all at least $2p$DO, and $v_{M-1,\bs,\bs}(\bh)$ is exactly $2p$DO.

Thus, $C_{M,\bs}(\bh)= C_M(\bs,\bs+\bh)$ is $2p$DO, and so the corresponding $M$-RA process $y_M(\cdot) \sim GP(0,C_M)$ is $p$ times (mean-square) differentiable at $\bs$.
\end{proof}

\begin{proof}[Proof of Proposition \ref{prop:continuity}]
First, note that realizations are (mean-square) continuous at $\bs \in \domain$, if $\lim_{\bh \rightarrow \bfzero} C_M(\bs,\bs+\bh) = C_M(\bs,\bs)$. Further, we have $\mu_M(\bs) = \E(y_M(\bs)|\bz) = \bz'\cov(\bz)^{-1}C_M(\locs,\bs)$. From the proof of Proposition \ref{prop:smoothness}, we have that $C_M(\bs_0,\bs+\bh)= \sum_{m=0}^M \sum_{j=1}^{r_m} a_{m,j}(\bs_0) v_m(\knot_{m,j},\bs+\bh)$. It is straightforward to show using a proof by induction very similar to that for Proposition \ref{prop:smoothness}, that $\lim_{\bh \rightarrow \bfzero} v_m(\knot_{m,j},\bs+\bh) = v_m(\knot_{m,j},\bs)$ if $\lim_{\bh \rightarrow \bfzero} \modu_m(\knot_{m,j},\bs+\bh) = \modu_m(\knot_{m,j},\bs)$ for all $m$. In contrast, if $\bs$ is on a region boundary, at least one $\modu_m(\knot_{m,j},\bs+\bh)$ will be discontinuous as a function of $\bh$, and so will $C_M(\bs_0,\bs+\bh)$ (unless $v_m(\bs,\bs+\bh)=w_m(\bs,\bs+\bh)$ and hence the $M$-RA-block is exact --- see Proposition \ref{prop:expexact}).
\end{proof}

\begin{lemma}[Sum of predictive processes]
\label{lem:sumpp}
For the decomposition in \eqref{eq:orthdecomp}, the sum of predictive processes up to resolution $m$ is equal in distribution to the predictive process based on the union of the knots up to resolution $m$, for any $m=0,1,\ldots,M$; that is, $\sum_{l=0}^{m}\tau_l(\cdot) \eqd E(y_0(\cdot)|y_0(\cup^{m}_{l=0}\knots_l))$.
\end{lemma}
\begin{proof}[Proof of Lemma \ref{lem:sumpp}]
 For $m=1$, $ \delta_1(\bs) \indep y_0(\knots_0)$, for any $\bs \in \domain$, because $E\big(\delta_1(\bs)y_0(\knots_0)\big)=E \Big( \big(y_0(\bs)-E(y_0(\bs)|y_0(\knots_0))\big)y_0(\knots_0) \Big)=E\big(y_0(\knots_0)\big)E(\delta_1(\bs))=0$, and $y_0(\knots_0)$, $\delta_1(\bs)$ are jointly Gaussian. And we have $ E(y_0(\cdot)|\delta_1(\knots_1),y_0(\knots_0)) = E(y_0(\cdot)|y_0(\knots_1),y_0(\knots_0))$,
 because for the $\sigma$-algebras
\[
\sigma(\delta_1(\knots_1),y_0(\knots_0))= \sigma \big( y_0(\knots_1)- E \big(y_0(\knots_1)|y_0(\knots_0)\big),y_0(\knots_0)  \big)= \sigma \big(y_0(\knots_1),y_0(\knots_0) \big),
\]
since $\sigma \big( y_0(\knots_1)- E \big(y_0(\knots_1)|y_0(\knots_0)\big),y_0(\knots_0)  \big) = \sigma \big( y_0(\knots_1)-f\big( y_0(\knots_0))\big),y_0(\knots_0) \big) \subset \sigma \big(y_0(\knots_1),y_0(\knots_0) \big)$, and the opposite also holds. Therefore,
\begin{align*} 
E \big(\delta_1(\bs)|\delta_1(\knots_1) \big) &  =  E \big( \delta_1(\bs)|\delta_1(\knots_1),y_0(\knots_0) \big) \\
  & = E \big( y_0(\bs)|\delta_1(\knots_1),y_0(\knots_0) \big) - E \big( E \big( y_0(\bs)|y_0(\knots_0) \big) |\delta_1(\knots_1),y_0(\knots_0) \big)\\
  & = E \big( y_0(\bs)|y_0(\knots_1),y_0(\knots_0) \big) - E \big( y_0(\bs)|y_0(\knots_0) \big),
\end{align*}
And so,
\[
\tau_0(\bs) + \tau_1(\bs) =   E \big(y_0(\bs)|y_0(\knots_0) \big) + E \big( \delta_1(\bs)|\delta_1(\knots_1) \big) =  E(y_0(\bs)|y_0(\knots_1),y_0(\knots_0)).
\]
Then, $\delta_2(\bs)=y_0(\bs)-E \big(y_0(\bs)|y_0(\knots_0 \cup \knots_1)\big)$, which implies $y_0(\knots_0 \cup \knots_1) \indep \delta_2(\bs)$. Iteratively repeat this argument to obtain $\sum_{l=0}^{m}\tau_l(\bs)=E \big(y_0(\bs)|y_0(\cup^{m}_{l=0}\knots_l) \big)$.
\end{proof}

\begin{lemma}[Block-independence for exponential covariance]
\label{lem:expind}
Assume $y_0(\cdot) \sim \GP(0,C_0)$, where $C_0$ is an exponential covariance function on the real line, $\domain=\mathbb{R}$, and consider a domain partitioning as in \eqref{eq:partitioning} with $r_m=(J-1) J^m$ knots for $m=0,\ldots,M-1$, which are placed such that at each resolution $m$ a knot is located on each boundary between two subregions at resolution $m+1$.
Then, for any $m=1,\ldots,M$, if $\bs_i \in \domain_\im$ and $\bs_j \in \domain_\jm$, we have $w_m(\bs_i,\bs_j)=0$ (defined in \eqref{eq:decompdetails}) if $(\im) \neq (\jm)$.
\end{lemma}
\begin{proof}[Proof of Lemma \ref{lem:expind}]
 For any $m=1,\ldots,M$, using Lemma \ref{lem:sumpp}, we have 
$$w_{m}(\bs_i,\bs_j)= C_0(\bs_i,\bs_j)-C_0(\bs_i,\knots^{m-1})C_0(\knots^{m-1},\knots^{m-1})^{-1}C_0(\knots^{m-1},\bs_j),$$ 
where $\knots^{m-1} \colonequals \cup^{m-1}_{l=0}\knots_l$. By the law of total covariance, 
\begin{align*}
w_m(\bs_i,\bs_j) & = C_0(\bs_i,\bs_j) -Cov \big( E \big(y_0(\bs_i)|y_0(\knots^{m-1}) \big),E\big(y_0(\bs_j)|y_0(\knots^{m-1})\big) \big)\\
 &= E \big(Cov \big(y_0(\bs_i),y_0(\bs_j)|y_0(\knots^{m-1})\big)\big).
\end{align*}
Because $(i_1,i_2,\ldots,i_{m-1}) \neq (j_1,j_2,\ldots,j_{m-1})$, there is a $\knot \in \knots^{m-1}$ that lies between $\bs_i$ and $\bs_j$. As $y_0(\cdot)$ is a Markov process \citep[e.g.,][Ch.~6]{Rasmussen2006}, $E \big( Cov \big(y_0(\bs_i),y_0(\bs_j)|y_0(\knots^{m-1})\big)\big)=  E \big( Cov \big(y_0(\bs_i),y_0(\bs_j)|y_0(\knot)\big)\big)= w_{m}(\bs_i,\bs_j)=0$.
\end{proof}

\begin{proof}[Proof of Proposition \ref{prop:expexact}]
Comparing the expression for $C_M$ in Lemma \ref{lem:covknot} to the expression for $C_0$ in \eqref{eq:c0decomp}, it is clear that $C_M(\bs_1,\bs_2)=C_0(\bs_1,\bs_2)$ if 
\begin{equation}
\label{eq:veqw}
v_m(\bs_i,\bs_j)=w_m(\bs_i,\bs_j), \quad \textnormal{for } m=0,\ldots,M \textnormal{ and any } \bs_i, \bs_j \in \domain.
\end{equation}
We now prove \eqref{eq:veqw} by induction. For $m=0$, we have $v_0(\bs_i,\bs_j)=C_0(\bs_i,\bs_j)\modu_0(\bs_i,\bs_j) = C_0(\bs_i,\bs_j)$, because $\modu_0(\bs_i,\bs_j) \equiv 1$ for the $M$-RA-block.
For $m>0$, assume that $v_{m-1}(\bs_i,\bs_j)=w_{m-1}(\bs_i,\bs_j)$. Then, we can write
\begin{equation}
\label{eq:vw2}
v_m(\bs_i,\bs_j)=w_m(\bs_i,\bs_j)\modu_m(\bs_i,\bs_j).
\end{equation}
Assume that $\bs_i \in \domain_\im$ and $\bs_j \in \domain_\jm$. Then, if $(\im) = (\jm)$, \eqref{eq:vw2} holds because $\modu_m(\bs_i,\bs_j)=1$. If $(\im) \neq (\jm)$, we have $\modu_m(\bs_i,\bs_j)=0$ but also $w_m(\bs_i,\bs_j) = 0$ by Lemma \ref{lem:expind}. This proves \eqref{eq:vw2}, which proves \eqref{eq:veqw}, which in turns proves Proposition \ref{prop:expexact}.
\end{proof}


\begin{proof}[Proof of Proposition \ref{prop:priorcomplexity}]
From \eqref{eq:priorw1}, we have $\bW_{m,l}^{k+1} = (\bW_{m,l}^{k} - \bX_{m,l}^k) \circ \modu_{k+1}(\knots_m,\knots_l)$, where $\bX_{m,l}^k \colonequals \bW_{m,k}^{k} \prec_{k}^{-1} \bW_{l,k}^{k}{}'$. The $(i,j)$th element of this matrix is
\begin{equation}
\label{eq:Vquad}
\textstyle (\bX_{m,l}^k)_{i,j} = \sum_{a,b=1}^{r_k} v_k(\knot_{m,i},\knot_{k,a})v_l(\knot_{l,j},\knot_{k,b}) (\prec_k^{-1})_{a,b},
\end{equation}
where $v_k(\knot_{m,i},\knot_{k,a}) =0$ if $\| \knot_{m,i} - \knot_{k,a} \| \geq d_k$, and $v_l(\knot_{l,j},\knot_{k,b})=0$ if $\| \knot_{l,j} - \knot_{k,b} \| \geq d_k$. Further, we only need the $(i,j)$th element of $\bW_{m,l}^{k+1}$ (and thus of $\bX_{m,l}^k$) if $(i,j) \in \mathcal{I}_{m,l}$, because $(\bW_{m,l}^l)_{i,j} =0$ if $\| \knot_{m,i} - \knot_{l,j} \| \geq d_l$. Hence, we only need $(\prec_k^{-1})_{a,b}$ if $\| \knot_{m,i} - \knot_{l,j} \| < d_l$, $\| \knot_{m,i} - \knot_{k,a} \| < d_k$, \emph{and} $\| \knot_{l,j} - \knot_{k,b} \| < d_k$, for some $m,l \in \{k+1,\ldots,M\}$. As $d_{k+1} = d_k/J > d_{k+2} > \ldots > d_M$, this means that do not need to calculate $(\prec_k^{-1})_{a,b}$ if $\|\knot_{k,a} - \knot_{k,b} \| \geq 2d_k + 2 d_{k+1} = (2+2/J)d_k$, and so we can replace $\prec_k^{-1}$ in $\bX_{m,l}^k$ by $\bS_k = \pprec_k^{-1} \circ \bG_k$.

Further, for each $(i,j) \in \mathcal{I}_{m,l}$, the time to compute \eqref{eq:Vquad} is $\order(r_0^2)$, because for any $\bs \in \domain$, the size of the set $\{\knot \in \knots_k: v_k(\bs,\knot) \neq 0 \}$ is $\order(r_0)$. As $\mathcal{I}_{m,l}$ is a set of size $\order(r_mr_0)$, the cost of computing $\bW_{m,l}^k$ for each $m,l,k$ is $\order(r_mr_0^3)$. Thus, the total computation time for $k=0,\ldots,l-1$, $l=0,\ldots,m$, and $m=0,\ldots,M$ is $\order(\sum_{m=0}^M \sum_{l=0}^m \sum_{k=0}^{l-1} r_m r_0^3) = \order(r_0^3 \sum_{m=0}^M r_m m^2) = \order( r_0^4 \sum_{m=0}^M J^m m^2) = \order(r_0^4 M^2 J^M) = \order(n M^2 r_0^3)$, because $n=\order(r_0J^M)$ and $\sum_{m=0}^M m^2 J^m \leq 2 M^2 J^M = \order(M J^M)$.
\end{proof}


\begin{proof}[Proof of Proposition \ref{prop:posteriornz}]
We have $(\pprec_{m,l})_{i,j} = 0$ if $\not\exists\, \bs \in \domain$ such that $\modu_m(\knot_{m,i},\bs) \neq 0$ and $\modu_l(\knot_{l,j},\bs)\neq 0$, or equivalently, if $\|\knot_{m,i} - \knot_{l,j}\| \geq d_m + d_l$. As $d_l = d_m J^{(l-m)/d}$, the $i$th row $(\pprec_{m,l})_{i,\cdot}$ has $\order(r_0 J^{(l-m)_+})$ nonzero elements, where $(x)_+ = x\mathbbm{1}_{\{x\geq0\}}$. The entire row of the matrix $\pprec$ corresponding to $\knot_{m,i}$ thus has $\order(r_0\sum_{l=0}^M J^{(l-m)_+}) = \order(r_0(m+J^{M-m}))$ nonzero elements. As there are $\order(r_0J^m)$ rows corresponding to resolution $m$, the total number of nonzero elements in $\pprec$ is $\order(\sum_{m=0}^M r_0J^m\cdot r_0(m+J^{M-m})) = \order(r_0^2(M J^M + \sum_{m=0}^M m J^m)) = \order(nMr_0)$, because $\sum_{m=0}^M m J^m \leq 2 M J^M = \order(M J^M)$ and $n=\order(r_0 J^M)$.
\end{proof}

\footnotesize
\bibliographystyle{apalike}
\bibliography{library}

\begin{thebibliography}{}

\bibitem[Ambikasaran et~al., 2016]{Ambikasaran2016}
Ambikasaran, S., Foreman-Mackey, D., Greengard, L., Hogg, D.~W., and O'Neil, M.
  (2016).
\newblock {Fast direct methods for Gaussian processes}.
\newblock {\em IEEE Transactions on Pattern Analysis and Machine Intelligence},
  38(2):252--265.

\bibitem[Banerjee et~al., 2008]{Banerjee2008}
Banerjee, S., Gelfand, A.~E., Finley, A.~O., and Sang, H. (2008).
\newblock {Gaussian predictive process models for large spatial data sets}.
\newblock {\em Journal of the Royal Statistical Society, Series B},
  70(4):825--848.

\bibitem[Chui, 1992]{Chui1992}
Chui, C. (1992).
\newblock {\em {An Introduction to Wavelets}}.
\newblock Academic Press.

\bibitem[Cressie and Johannesson, 2008]{Cressie2008}
Cressie, N. and Johannesson, G. (2008).
\newblock {Fixed rank kriging for very large spatial data sets}.
\newblock {\em Journal of the Royal Statistical Society, Series B},
  70(1):209--226.

\bibitem[Datta et~al., 2016]{Datta2016}
Datta, A., Banerjee, S., Finley, A.~O., and Gelfand, A.~E. (2016).
\newblock {Hierarchical Nearest-Neighbor Gaussian Process Models for Large
  Geostatistical Datasets}.
\newblock {\em Journal of the American Statistical Association},
  111(514):800--812.

\bibitem[Erisman and Tinney, 1975]{Erisman1975}
Erisman, A.~M. and Tinney, W.~F. (1975).
\newblock {On computing certain elements of the inverse of a sparse matrix}.
\newblock {\em Communications of the ACM}, 18(3):177--179.

\bibitem[Furrer et~al., 2006]{Furrer2006}
Furrer, R., Genton, M.~G., and Nychka, D.~W. (2006).
\newblock {Covariance tapering for interpolation of large spatial datasets}.
\newblock {\em Journal of Computational and Graphical Statistics},
  15(3):502--523.

\bibitem[Gelfand et~al., 2010]{Gelfand2010}
Gelfand, A., Diggle, P., Guttorp, P., and Fuentes, M., editors (2010).
\newblock {\em {Handbook of Spatial Statistics}}.
\newblock CRC Press.

\bibitem[Gneiting, 2002]{Gneiting2002}
Gneiting, T. (2002).
\newblock {Compactly supported correlation functions}.
\newblock {\em Journal of Multivariate Analysis}, 83(2):493--508.

\bibitem[Gneiting and Katzfuss, 2014]{Gneiting2014}
Gneiting, T. and Katzfuss, M. (2014).
\newblock {Probabilistic forecasting}.
\newblock {\em Annual Review of Statistics and Its Application}, 1(1):125--151.

\bibitem[Heaton et~al., 2017]{Heaton2017}
Heaton, M.~J., Datta, A., Finley, A., Furrer, R., Guhaniyogi, R., Gerber, F.,
  Gramacy, R.~B., Hammerling, D., Katzfuss, M., Lindgren, F., Nychka, D.~W.,
  Sun, F., and Zammit-Mangion, A. (2017).
\newblock {Methods for analyzing large spatial data: A review and comparison}.
\newblock {\em arXiv:1710.05013}.

\bibitem[Higdon, 1998]{Higdon1998}
Higdon, D. (1998).
\newblock {A process-convolution approach to modelling temperatures in the
  North Atlantic Ocean}.
\newblock {\em Environmental and Ecological Statistics}, 5(2):173--190.

\bibitem[Johannesson et~al., 2007]{Johannesson2007}
Johannesson, G., Cressie, N., and Huang, H.-C. (2007).
\newblock {Dynamic multi-resolution spatial models}.
\newblock {\em Environmental and Ecological Statistics}, 14(1):5--25.

\bibitem[Kanter, 1997]{Kanter1997}
Kanter, M. (1997).
\newblock {Unimodal spectral windows}.
\newblock {\em Statistics {\&} Probability Letters}, 34(4):403--411.

\bibitem[Katzfuss, 2013]{Katzfuss2012}
Katzfuss, M. (2013).
\newblock {Bayesian nonstationary spatial modeling for very large datasets}.
\newblock {\em Environmetrics}, 24(3):189--200.

\bibitem[Katzfuss, 2017]{Katzfuss2015}
Katzfuss, M. (2017).
\newblock {A multi-resolution approximation for massive spatial datasets}.
\newblock {\em Journal of the American Statistical Association},
  112(517):201--214.

\bibitem[Katzfuss and Cressie, 2009]{Katzfuss2009}
Katzfuss, M. and Cressie, N. (2009).
\newblock {Maximum likelihood estimation of covariance parameters in the
  spatial-random-effects model}.
\newblock In {\em Proceedings of the Joint Statistical Meetings}, pages
  3378--3390, Alexandria, VA. American Statistical Association.

\bibitem[Katzfuss and Cressie, 2011]{Katzfuss2010}
Katzfuss, M. and Cressie, N. (2011).
\newblock {Spatio-temporal smoothing and EM estimation for massive
  remote-sensing data sets}.
\newblock {\em Journal of Time Series Analysis}, 32(4):430--446.

\bibitem[Katzfuss and Cressie, 2012]{Katzfuss2011}
Katzfuss, M. and Cressie, N. (2012).
\newblock {Bayesian hierarchical spatio-temporal smoothing for very large
  datasets}.
\newblock {\em Environmetrics}, 23(1):94--107.

\bibitem[Katzfuss and Guinness, 2017]{Katzfuss2017a}
Katzfuss, M. and Guinness, J. (2017).
\newblock {A general framework for Vecchia approximations of Gaussian
  processes}.
\newblock {\em arXiv:1708.06302}.

\bibitem[Katzfuss and Hammerling, 2017]{Katzfuss2014}
Katzfuss, M. and Hammerling, D. (2017).
\newblock {Parallel inference for massive distributed spatial data using
  low-rank models}.
\newblock {\em Statistics and Computing}, 27(2):363--375.

\bibitem[Kaufman et~al., 2008]{Kaufman2008}
Kaufman, C.~G., Schervish, M., and Nychka, D.~W. (2008).
\newblock {Covariance tapering for likelihood-based estimation in large spatial
  data sets}.
\newblock {\em Journal of the American Statistical Association},
  103(484):1545--1555.

\bibitem[Li et~al., 2008]{Li2008}
Li, S., Ahmed, S., Klimeck, G., and Darve, E. (2008).
\newblock {Computing entries of the inverse of a sparse matrix using the FIND
  algorithm}.
\newblock {\em Journal of Computational Physics}, 227(22):9408--9427.

\bibitem[Lin et~al., 2011]{Lin2011}
Lin, L., Yang, C., Meza, J., Lu, J., Ying, L., and Weinan, E. (2011).
\newblock {SelInv - An algorithm for selected inversion of a sparse symmetric
  matrix}.
\newblock {\em ACM Transactions on Mathematical Software}, 37(4):40.

\bibitem[Lindgren et~al., 2011]{Lindgren2011a}
Lindgren, F., Rue, H., and Lindstr{\"{o}}m, J. (2011).
\newblock {An explicit link between Gaussian fields and Gaussian Markov random
  fields: the stochastic partial differential equation approach}.
\newblock {\em Journal of the Royal Statistical Society, Series B},
  73(4):423--498.

\bibitem[Mardia et~al., 1998]{Mardia1998}
Mardia, K., Goodall, C., Redfern, E., and Alonso, F. (1998).
\newblock {The kriged Kalman filter}.
\newblock {\em Test}, 7(2):217--282.

\bibitem[McLeod et~al., 2007]{McLeod2007}
McLeod, A.~I., Yu, H., and Krougly, Z. (2007).
\newblock {Algorithms for linear time series analysis: With R package}.
\newblock {\em Journal of Statistical Software}, 23(5).

\bibitem[Nychka et~al., 2015]{Nychka2012}
Nychka, D.~W., Bandyopadhyay, S., Hammerling, D., Lindgren, F., and Sain, S.~R.
  (2015).
\newblock {A multi-resolution Gaussian process model for the analysis of large
  spatial data sets}.
\newblock {\em Journal of Computational and Graphical Statistics},
  24(2):579--599.

\bibitem[Qui{\~{n}}onero-Candela and Rasmussen, 2005]{Quinonero-Candela2005}
Qui{\~{n}}onero-Candela, J. and Rasmussen, C. (2005).
\newblock {A unifying view of sparse approximate Gaussian process regression}.
\newblock {\em Journal of Machine Learning Research}, 6:1939--1959.

\bibitem[Rasmussen and Williams, 2006]{Rasmussen2006}
Rasmussen, C.~E. and Williams, C. K.~I. (2006).
\newblock {\em {Gaussian Processes for Machine Learning}}.
\newblock MIT Press.

\bibitem[Sang and Huang, 2012]{Sang2012}
Sang, H. and Huang, J.~Z. (2012).
\newblock {A full scale approximation of covariance functions}.
\newblock {\em Journal of the Royal Statistical Society, Series B},
  74(1):111--132.

\bibitem[Sang et~al., 2011]{Sang2011a}
Sang, H., Jun, M., and Huang, J.~Z. (2011).
\newblock {Covariance approximation for large multivariate spatial datasets
  with an application to multiple climate model errors}.
\newblock {\em Annals of Applied Statistics}, 5(4):2519--2548.

\bibitem[Snelson and Ghahramani, 2007]{Snelson2007}
Snelson, E. and Ghahramani, Z. (2007).
\newblock {Local and global sparse Gaussian process approximations}.
\newblock In {\em Artificial Intelligence and Statistics 11 (AISTATS)}.

\bibitem[Stein, 1999]{Stein1999}
Stein, M.~L. (1999).
\newblock {\em {Interpolation of Spatial Data: Some Theory for Kriging}}.
\newblock Springer, New York, NY.

\bibitem[Stein, 2011]{Stein2011b}
Stein, M.~L. (2011).
\newblock {2010 Rietz lecture: When does the screening effect hold?}
\newblock {\em Annals of Statistics}, 39(6):2795--2819.

\bibitem[Stein et~al., 2004]{Stein2004}
Stein, M.~L., Chi, Z., and Welty, L. (2004).
\newblock {Approximating likelihoods for large spatial data sets}.
\newblock {\em Journal of the Royal Statistical Society: Series B},
  66(2):275--296.

\bibitem[Vaidya, 1989]{Vaidya1989}
Vaidya, P.~M. (1989).
\newblock {An O(n log n) algorithm for the all-nearest-neighbors problem}.
\newblock {\em Discrete \& Computational Geometry}, 4(2):101--115.

\bibitem[Vecchia, 1988]{Vecchia1988}
Vecchia, A. (1988).
\newblock {Estimation and model identification for continuous spatial
  processes}.
\newblock {\em Journal of the Royal Statistical Society, Series B},
  50(2):297--312.

\bibitem[Wikle and Cressie, 1999]{Wikle1999}
Wikle, C.~K. and Cressie, N. (1999).
\newblock {A dimension-reduced approach to space-time Kalman filtering}.
\newblock {\em Biometrika}, 86(4):815--829.

\end{thebibliography}

\end{document}